\documentclass[onecolumn, twoside]{IEEEtran}
\IEEEoverridecommandlockouts

\usepackage{ifpdf}
\usepackage{cite}
\usepackage[cmex10]{amsmath}
\usepackage{enumerate}
\usepackage{amsthm}
\usepackage{latexsym}
\usepackage{amssymb}
\usepackage{xspace}
\usepackage{amscd}
\usepackage{enumerate}
\usepackage{graphicx}
\usepackage{epic}

\newcommand{\N}{{\mathbb N}}
\newcommand{\A}{\mathbb A}

\newcommand{\Z}{{\mathbb Z}}
\newcommand{\F}{\mathbb F}

\newcommand{\D}{{\mathcal D}}

\newcommand{\bj}{{\textbf j}}
\newcommand{\C}{{\mathcal C}}

\newcommand{\Zc}{{\mathcal Z}}
\newcommand{\I}{{\mathcal I}}
\newcommand{\U}{{\mathcal U}}

\newcommand{\tq}{\; \mid \;}
\newcommand{\doble}[2]{\genfrac{}{}{0cm}{2}{#1}{#2}}

\newcommand{\topeK}[1]{\left\lfloor\frac{#1}{K}\right\rfloor}

\newcommand{\Db}{\overline{\D^*}}
\newcommand{\Or}{\mathcal{O}}

\usepackage{algorithmic}
\usepackage{array}

\usepackage{mdwmath}
\usepackage{mdwtab}
\usepackage{fixltx2e}

\usepackage{stfloats}
\usepackage{color}

\newtheorem{theorem}{Theorem}
\newtheorem{lemma}[theorem]{Lemma}
\newtheorem{proposition}[theorem]{Proposition}
\newtheorem{definition}[theorem]{Definition}

\newtheorem{remark}[theorem]{Remark}

\newtheorem{example}[theorem]{Example}
\newtheorem{examples}[theorem]{Examples}

\begin{document}

%_
% paper title
% can use linebreaks \\ within to get better formatting as desired
\title{Information sets from defining sets for Reed-Muller codes of first and second order  
\thanks{This work was partially supported by MINECO, project MTM2016-77445-P, and Fundaci\'{o}n S\'{e}neca of Murcia, project 19880/GERM/15.
An extended abstract of a part of this paper was presented at Fith International Castle Meeting on Coding Theory and Applications, Estonia, 2017}}

% author names and affiliations
% use a multiple column layout for up to three different
% affiliations
% \author{\IEEEauthorblockN{Diana H. Bueno-Carre\~no$^\star$}
% \IEEEauthorblockA{Departamento de CN y Matemáticas\\
% Pontificia Universidad Javeriana-Cali\\
% Cali, Colombia\\
% Email: dhbueno@javerianacali.edu.co}
% \and
% \IEEEauthorblockN{Jos\'e Joaqu\'{i}n Bernal \\ and \\ Juan Jacobo Sim\'on}
% \IEEEauthorblockA{Departamento de Matemáticas\\
% Universidad de Murcia\\
% 30100 Murcia, Spain.\\
% Email: josejoaquin.bernal@um.es and jsimon@um.es} }

% conference papers do not typically use \thanks and this command
% is locked out in conference mode. If really needed, such as for
% the acknowledgment of grants, issue a \IEEEoverridecommandlockouts
% after \documentclass

% for over three affiliations, or if they all won't fit within the width
% of the page, use this alternative format:
% 
\author{\IEEEauthorblockN{Jos\'e Joaqu\'{i}n Bernal
 and
Juan Jacobo Sim\'on (Member, IEEE). \\
\IEEEauthorblockA{Departamento de Matem\'{a}ticas\\
Universidad de Murcia,
30100 Murcia, Spain.\\ Email: \{josejoaquin.bernal, jsimon\}@um.es} \\
%\IEEEauthorblockA{\IEEEauthorrefmark{2}Departamento de Ciencias Naturales y Matem\'{a}ticas\\
%Pontificia Universidad Javeriana seccional Cali, Colombia\\
% Email: dhbueno@javerianacali.edu.co}
}}

% make the title area
\maketitle

\begin{abstract}
 Reed-Muller codes belong to the family of affine-invariant codes. As such codes they have a defining set that determines them uniquely, and they are extensions of cyclic group codes. In this paper we identify those cyclic codes with multidimensional abelian codes and we use the techniques introduced in \cite{BS} to construct information sets for them from their defining set. For first and second order Reed-Muller codes, we describe a direct method to construct information sets in terms of their basic parameters.
\end{abstract}

% For peer review papers, you can put extra information on the cover
% page as needed:
% \ifCLASSOPTIONpeerreview
% \begin{center} \bfseries EDICS Category: 3-BBND \end{center}
% \fi
%
% For peerreview papers, this IEEEtran command inserts a page break and
% creates the second title. It will be ignored for other modes.
\IEEEpeerreviewmaketitle

\section{Introduction}

The family of Reed-Muller codes was introduced by D. E. Muller in 1954 \cite{Muller} and a specific decoding algorithm for them was presented by I. S. Reed in the same year \cite{Reed}. Since then, many authors have paid attention to this family of codes for many reasons. On the one hand, they can be implemented and decoded easily, and on the other hand they have a rich algebraic structure which allows us to see them from different points of view. Originally, Reed-Muller codes were defined in terms of boolean functions by Muller and, equivalently, from this point of view they can be constructed as polynomial codes. In another context, they can be treated as geometric codes; specifically, they can be identified with codes of the designs of points and flats in the affine space over the binary field. Finally, they also possess an algebraic structure that let us treat them as group algebra codes. We are interested in this last point of view. Specifically, any Reed-Muller code can be identified with an ideal in a group algebra of an elementary abelian group. Moreover, the family of Reed-Muller codes is contained in the family of so-called affine-invariant codes and so a defining set can be defined for them.  The reader may see \cite{Hb2} for a comprehensive explanation on the structures of Reed-Muller codes. 

In particular, we are interested in the problem of finding information sets for Reed-Muller codes, a question that  has been addressed for many authors earlier. 
%The interest of this paper is to describe a new method for constructing information sets for a Reed-Muller code by using its algebraic structure. 
From the geometric point of view several ideas for finding information sets has been presented. In \cite{BlokMoor} and \cite{Moor} Moorhouse and Blokhuis gave bases formed by the incidence vectors of certain lines valid for a more general family of geometric codes. Later, J. D. Key, T. P. McDonough and V. C. Mavron extended these definitions in \cite{KMM} in order to apply the permutation decoding algorithm. Finally, in \cite{KMM2}, the same authors gave a simple description of information sets for Reed-Muller codes by using the polynomial approach. 

In this work we use the fact that Reed-Muller codes can be seen as extended-cyclic affine-invariant codes to get information sets. From this context any Reed-Muller code is a parity check extension of a cyclic group code, so any information set of that cyclic code is obviously an information set for the Reed-Muller code; moreover, there exists a direct connection between the respective defining sets. On the other hand, in \cite{BS} we introduced a method for constructing information sets for any abelian code starting from its defining set. Then, the goal of this paper is to obtain information sets for Reed-Muller codes of first and second order respectively 
by applying those techniques shown in \cite{BS} to the punctured cyclic group code seen as a multidimensional abelian code.

%apply the method introduced in [] for abelian codes to Reed-Muller codes. More precisely, we start from the description of the defining set of a Reed-Muller code, seen as an affine-invariant code, and we use those multidimensional techniques to obtain an information set. 

The paper is structured as follows. Section \ref{Preliminaries} includes the basic notation and the necessary preliminaries about abelian codes. In Section \ref{InfoSets} we recall how the method given in \cite{BS} works in the particular case of two-dimensional abelian codes. In Section \ref{cyclics} we prove some necessary results related to the case of cyclic codes seen as two-dimensional cyclic codes; these results will be essential in the rest of the paper. In Section \ref{RM}  we focus on Reed-Muller codes: we recall their definition as affine-invariant codes in an abelian group algebra and the description of the defining sets. Then we show how to apply the results in Section \ref{cyclics} to Reed-Muller codes under some restrictions on their parameters. Section \ref{first-order} and Secion \ref{second-order} are devoted to develop in detail the construction of information sets for first-order and second-order Reed-Muller codes respectively. In the case of first-order Reed-Muller codes we get a description of the information sets directly from the previous sections; for second-order Reed-Muller codes, although the development is more complicated, involving several technical results, it is remarkable that again the description turns out to be particularly simple in the end. %Finally, in Section \ref{second-order} we include the construction for second-order Reed-Muller codes and we comment how it becomes more difficult as the order rises.

%Section \ref{a=2} is devoted to develop in detail the construction of information sets in the case of Reed-Muller codes with length $2^m-1$ where $m$ is even. Finally, we include in Section \ref{m-rho=2} our work concerning the description of information sets in the case of Reed-Muller codes of length $2^m-1$ an order $m-2$. In this case the results turn out to be particularly simple. 

\section{Preliminaries}\label{Preliminaries}

In this paper we deal with Reed-Muller codes identified as abelian codes, so for the convenience of the reader %and for the sake of simplicity 
we give an introduction to abelian codes just in the binary case. Then all throughout $\F$ denotes the field with two elements.

A binary abelian code is an ideal of a group algebra $\F G$, where $G$ is an abelian group. It is well-known that there exist integers $r_1,\dots,r_l$ such that $G$ is isomorphic to the direct product $C_{r_1}\times\cdots\times C_{r_l}$, with $C_{r_i}$ the cyclic group of order $r_i$, $i=1,\dots,l$. Moreover, this decomposition yields an isomorphism of $\F$-algebras from $\F G$ to 
  $$\F[X_1,\dots,X_l]/\left\langle X_1^{r_1}-1,\dots,X_l^{r_l}-1\right\rangle.$$ 
 We denote this quotient algebra by $\A(r_1,\dots,r_l)$ and we identify the codewords with polynomials $P(X_1,\dots,X_l)$ such that every monomial satisfies that the degree of the indeterminate $X_i$ is in $\Z_{r_i}$, the ring of integers modulo $r_i$, and that we always write as canonical representatives (that is, non negative integers less than $r_i$). We  write the elements $P\in \A(r_1,\dots,r_l)$ as $P=P(X_1,\dots,X_l)=\sum a_\bj X^\bj$, where $\bj=(j_1,\dots, j_l)\in \Z_{r_1}\times\dots\times\Z_{r_l}$, $X^\bj=X_1^{j_1}\cdots X_l^{j_l}$ and $a_\bj\in\F$. We always assume that $r_i$ is odd for every $i=1,\dots,l$, that is, we assume that $\A(r_1,\dots,r_l)$ is a semisimple algebra.

Our main tool to study the construction of information sets for abelian codes is the notion of defining set.

\begin{definition}\label{definingset}
Let $\C\subseteq \A(r_1,\dots,r_l)$ be an abelian code. Let $R_i$ be the set of $r_i$-th roots of unity, $i=1,\dots,l$. Then the \textit{root} set of $\C$ is given by
	$$\Zc(\C)=\left\{(\beta_1,\dots,\beta_l)\in\prod_{i=1}^l R_i \tq  P(\beta_1,\dots,\beta_l)=0 \text{ for all } P(X_1,\dots,X_l)\in \C\right\}.$$

 Then, for a fixed primitive $r_i$-th root of unity $\alpha_i$ in some extension of $\F$,  $i=1,\dots,l$, the \textit{defining} set of $\C$ with respect to $\alpha=\{\alpha_1,\dots,\alpha_l\}$ is
$$ D_\alpha\left(\C\right)=\left\{ (a_1,\dots,a_l)\in  \Z_{r_1}\times\dots\times\Z_{r_l} \tq  (\alpha_1^{a_1},\dots,\alpha_l^{a_l}) \in \Zc(\C) \right\}.$$
\end{definition}
 It can be proved that, fixed a collection of primitive roots of unity, every abelian code is totally determined by its defining set.

% Note also that given an abelian code $\C\subseteq \A(r_1,\dots,r_l)$ with defining set $D_\alpha(\C)$  if one chooses different primitive roots of unity, say $\gamma=\{\gamma_1,\dots,\gamma_l\}$, then the set $D_\gamma(\C)$ detemines a new code, say $\C'$, such that $\D_\alpha(\C')=\D_\gamma(\C)$, which is equivalent to $\C$ in the sense that there exists a permutation that transforms $\C$ in $\C'$. So, for the sake of brevity, we refer to abelian codes without any mention to the primitive roots that we are using as reference, and we denote the defining set of $\C$ by $D(\C)$.

\begin{remark}
 The reader may check that in the case $l=1$ the previous definitions coincide with the classical notions of zeros and defining set of a cyclic code.
\end{remark}

In order to describe the structure of the defining set of an abelian code we need to introduce the following definitions.

\begin{definition}
 Let $a,r$ and $\gamma$ be integers. The $2^{\gamma}$-cyclotomic coset of $a$ modulo $r$ is the set 
\[ C_{2^\gamma,r}(a)=\left\{a\cdot 2^{\gamma\cdot i} \tq i \in \N\right\} \subseteq \Z_r.\]
\end{definition}

We shall write $C_r(a)$ when $\gamma=1$.

\begin{definition}\label{qorbita}
Given $(a_1,\dots,a_l)\in \Z_{r_1}\times\dots\times\Z_{r_l}$, its \textit{$2$-orbit} modulo  $\left(r_1,\ldots,r_l \right)$ is the set
$$ Q(a_1,\dots,a_l)=\left\{\left(a_1\cdot 2^{i} ,\dots, a_l\cdot 2^{i}  \right)\tq i\in \N\right\} \subseteq \Z_{r_1}\times\dots\times\Z_{r_l}.$$
\end{definition}
	
It is well known that for every abelian code $\C\subseteq\A(r_1,\dots,r_l)$, $D\left(\C\right)$ is closed under multiplication by $2$ in $\Z_{r_1}\times\dots\times\Z_{r_l}$, and so $D(\C)$ is a disjoint union of $2$-orbits modulo $(r_1,\dots,r_l)$. Conversely, every union of $2$-orbits modulo $(r_1,\dots,r_l)$ defines an abelian code in $\A(r_1,\dots,r_l)$. From now on, we will only write $2$-orbit, and the tuple of integers will always be clear by the context. 

To finish this section we give the notion of information set of a code in the context of abelian codes. For $P=\sum a_\bj X^\bj\in \A(r_1,\dots,r_l)$ and $\I$ be a subset of $\Z_{r_1}\times\cdots\times\Z_{r_l}$, we denote by $P_\I$ the vector $(a_\bj)_{\bj\in\I}\in \F^{\left|\I\right|}$. Now, for an abelian code $\C \subseteq \A(r_1,\dots,r_l)$ we denote by $\C_\I$ the linear code $\{P_\I:P\in\C\}\subseteq \F^{\left|\I\right|}$.

\begin{definition}
 An information set for an abelian code $\C\subseteq\A(r_1,\dots,r_l)$ with dimension $k$ is a set $\I\subseteq \Z_{r_1}\times\cdots\times\Z_{r_l}$ such that $\left|\I\right|=k$ and $\C_\I=\F^k$. 
 
 The complementary set $\left(\Z_{r_1}\times\cdots\times\Z_{r_l}\right)\setminus \I$ is called a set of check positions for $\C$.
\end{definition}

As usual, we denote by $\C^\bot$ the dual code of $\C$, that is, the set of codewords $v\in \A(r_1,\dots,r_l)$ such that $v\cdot u=0$, for all $u\in \C$, where $``\cdot"$ denotes the usual inner product. It is easy to see that any information set for $\C$ is a set of check positions for $\C^\bot$ and vice versa.

\section{Information sets for abelian codes}\label{InfoSets}

In \cite{BS} we introduced a method for constructing information sets for any multidimensional abelian code just in terms of its defining set. In this section we only recall the two-dimensional construction because it is the unique case that we will use. So, from now on we take $l=2$ and the ambient space will be $\A(r_1,r_2)$.

Let $e=(e_1,e_2)\in \Z_{r_1}\times\Z_{r_2}$. We define
\begin{equation*}
 m(e_1)=\left|C_{r_1}(e_1)\right|
\end{equation*}
and
\begin{equation}\label{parametrosm}
 m(e)=m(e_1,e_2)=\left|C_{2^{m(e_1)},r_2}(e_2)\right|.
\end{equation}

The construction of information sets is based on the computation of the parameters defined in (\ref{parametrosm}) on a special subset of the defining set of the given abelian code. Specifically this set has to satisfy the conditions described in the following definition. For any $A\subset\Z_{r_1}\times\Z_{r_2}$ we denote its projection onto the first coordinate by $A_1$.
	\begin{definition}\label{restricted representatives}
Let $D$ be a union of $2$-orbits modulo $(r_1,r_2)$ and $\overline{D}\subset D$ a complete set of representatives. Then $\overline D$ is called a set of restricted representatives if $\overline{D}_1$ is a complete set of representatives of the $2$-cyclotomic cosets modulo $r_1$ in $D_1$.
 \end{definition}
 
\begin{remark}
  In \cite{BS} we gave the notion of restricted representatives associated to a fixed ordering on the indeterminates $X_1,X_2$. In the two-dimensional case that definition is equivalent to Definition~\ref{restricted representatives} when we fix the ordering $X_1<X_2$. This will be our default ordering, so we will make no reference to the order on the indeterminates in the rest of the paper. 
\end{remark}

\begin{example}\label{ejemplo restricted}
 Consider $r_1=3, r_2=5$ and let $\C\subseteq\A(3,5)$ be the abelian code with defining set $\D(\C)=Q(1,1)\cup Q(1,2)\cup Q(0,0)$, where $Q(1,1)=\{(1,1),(2,2),(1,4),(2,3)\}$, $Q(1,2)=\{(1,2),(2,4),(1,3),(2,1)\}$ and $Q(0,0)=\{(0,0)\}$. Then, according to Definition~\ref{restricted representatives}, the set of representatives $\{(1,1),(2,1),(0,0)\}$ is not restricted, because $C_3(1)=C_3(2)$, while $\{(1,1), (1,2), (0,0)\}$ is indeed restricted. 
\end{example}
  
  Now, let $\C\subseteq\A(r_1,r_2)$ be an abelian code with defining set $D_\alpha\subseteq\Z_{r_1}\times\Z_{r_2}$, with respect to $\alpha=\{\alpha_1,\alpha_2\}$. Take $\overline D\subset D_\alpha$ a set of restricted representatives. %As above, we denote by $\overline D_1$ the projection of $\overline D$ onto $\Z_{r_1}$. 
  Given $e_1\in \overline D_1$, let
\begin{equation*}\label{R(e)}
 R(e_1)=\{e_2\in \Z_{r_{2}}\mid (e_1,e_2)\in \overline D\}.
\end{equation*} 

For each $e_1\in \overline D_{1}$, we define
	\begin{equation}\label{M(e)}
	 M(e_1)=\sum_{e_2\in R(e_1)} m\left(e_1,e_2\right)
	\end{equation} 
and we consider the values $\{M(e_1)\}_{e_1\in \overline D_{1}}$. Then we denote 
\begin{eqnarray*}
    f_1&=&\max\limits_{e_1\in \overline D_{1}}\{M(e_1)\} \qquad\text{ and }\\
    f_i&=&\max\limits_{e_1\in \overline D_{1}}\{M(e_1)\mid M(e_1)<f_{i-1}\}.
   \end{eqnarray*}
   
So, we obtain the sequence
   \begin{equation}\label{sequence f[i]}
    f_{1}>\dots>f_{s}>0=f_{s+1}, 
   \end{equation} 
that is, we denote by $f_{s}$ the minimum value of the parameters $M(\cdot)$ and we set $f_{s+1}=0$ by convention. Note that $M(e_1)>0$, for all $e_1\in \overline D_{1}$, by definition.

From the previous values $f_i$ we define for $i=1,\dots,s$
  \begin{equation}\label{lasg}
	g_{i}=\sum\limits_{M(e_1)\geq f_i}m(e_1) 
	\end{equation}
and then we obtain the sequence
  \begin{equation}\label{sequence gi}
    g_1<g_2<\cdots<g_s.
  \end{equation}  

Finally, we define the set 

  \begin{equation}\label{checkpositions}
 \Gamma(\C)=\{(i_1,i_2)\in \Z_{r_1}\times\Z_{r_2}\mid  \text{ there exists } 1\leq j\leq s \text{ with } f_{j+1}\leq i_2< f_{j}, \text{ and } 0\leq i_1<g_{j}\}.
\end{equation}

The following theorem, proved in \cite{BS} for any abelian code, establishes that $\Gamma(\C)$ is a set of check positions for $\C$, and consequently $\Gamma(\C)$ defines an information set for $\C^\bot$.

 \begin{theorem}\label{teorema principal check positions}
   Let $r_1, r_2$ be odd integers and let $\C$ be an abelian code in  $\A(r_1,r_2)$ with defining set $\D_\alpha(\C)$ with respect to $\alpha=\{\alpha_1,\alpha_2\}$. Then $\Gamma(\C)$ is a set of check positions for $\C$.
  \end{theorem}
  
 Let us observe that given an abelian code $\C\subseteq \A(r_1,r_2)$ with defining set $D_\alpha(\C)$  if one chooses different primitive roots of unity, say $\gamma=\{\gamma_1,\gamma_2\}$, then the \textit{structure} of the $q$-orbits in $D_\gamma(\C)$ is the same as in $D_\alpha(\C)$, that is, we obtain the same values for the parameters (\ref{parametrosm}) and (\ref{M(e)}), so we get the same set of check positions $\Gamma(\C)$ (see \cite[p.100]{Th}). That is why we do not use any reference to the roots of unity taken in the notation of the set of check positions. So, in the rest of the paper, for any abelian code $\C$ we denote its defining set by $D(\C)$ and the corresponding set of check positions by $\Gamma(\C)$ without any mention of the primitive roots.

\begin{example}
 We continue with the code $\C$ considered in Example~\ref{ejemplo restricted}. Then $\D(\C)=Q(1,1)\cup Q(1,2)\cup Q(0,0)$ and we take $\overline\D=\{(1,1),(1,2),(0,0)\}$ as set of restricted representatives. From (\ref{M(e)}) we have that $M(1)=m(1,1)+m(1,2)=2+2=4$ and $M(0)=m(0,0)=1$. So $f_1=4> f_2=1>f_3=0$. On the other hand, from (\ref{lasg}) we obtain $g_1=m(1)=2, g_2=g_1+m(0)=3$. Therefore, 
 \begin{eqnarray*}
  \Gamma(\C)&=&\{(i_1,i_2)\in \Z_{r_1}\times\Z_{r_2}\mid  (1\leq i_2< 4 \text{ and } 0\leq i_1<2) \text { or } (0\leq i_2< 1 \text{ and } 0\leq i_1<3)\}\\
  &=&\{(0,0),(0,1),(0,2),(0,3),(1,0),(1,1),(1,2),(1,3),(2,0)\}
 \end{eqnarray*}
 is a set of check positions for $\C$.
\end{example}
 
\begin{remark}
 In \cite{BS} we showed how to construct the previous set of check positions for any multidimensional abelian code. As we have already mentioned, in this paper we only need to use the two-dimensional case which yields the same information set that was introduced by H. Imai in \cite{Imai}.
\end{remark}

\section{Cyclic codes as two-dimensional cyclic codes}\label{cyclics}
 
% Our target is the construction of information sets for a Reed-Muller code by using the method showed in the previous section. 
We are going to construct an information set for the punctured cyclic code of a Reed-Muller code viewed as a multidimensional abelian code. So, we are interested in applying the results of the previous section when the original abelian code is in fact cyclic. 

Let $\C^*$ be a binary cyclic code with length $n=r_1\cdot r_2$. All throughout this section we assume that $\gcd(r_1,r_2)=1$ and $r_1,r_2$ are odd. Then let $T:\Z_n\rightarrow \Z_{r_1}\times \Z_{r_2}$ be an isomorphism, and let us denote $T=(T_1,T_2)$; that is, $T_i(e)$ is the projection of $T(e)$ onto $\Z_{r_i}$, for $i=1,2$ and any $e\in \Z_n$. 

The proof of the following result is straightforward.

\begin{lemma}\label{T-orbitas}
  For every $e\in \Z_n$ the equality $T\left(C_n(e)\right)=Q\left(T(e)\right)$ holds; in particular $\left|C_n(e)\right|=\left|Q(T(e))\right|$. Moreover, one has that the projection of $Q(T(e))$ onto $\Z_{r_1}$ is equal to $C_{r_1}(T_1(e))$.
\end{lemma}

Let $\D^*=D_\alpha(\C^*)\subseteq \Z_n$ be the defining set of $\C^*$ with respect to an arbitrary primitive $n$-th  root of unity $\alpha$. Then, since there exist integers $\eta_1,\eta_2$ such that $\eta_1 r_1+\eta_2 r_2=1$, we have that $\alpha_1=\alpha^{\eta_2 r_2}$ and $\alpha_2=\alpha^{\eta_1 r_1}$ are primitive $r_1$-th and $r_2$-th roots of unity respectively. Fix $T$ an isomorphism as above and set $T(1)=(\delta_1,\delta_2)$; observe that $\gcd(\delta_1,r_1)=1$ and $\gcd(\delta_2,r_2)=1$. We define the abelian code $\C=\C_{(\C^*,T)}\subseteq \A(r_1,r_2)$ as the code with defining set $\D=D(\C)=T(\D^*)$, with respect to $(\beta_1,\beta_2)=(\alpha_1^{\delta_1^{-1}},\alpha_2^{\delta_2^{-1}})$. In this situation, we have that $\C$ is the image of $\C^*$ by the map 
  \begin{eqnarray*}
  \nonumber \A(n) &\longrightarrow& \A(r_1,r_2)\\
   \sum a_{i}X^i&\hookrightarrow& \sum b_{jl}X^jY^l,
  \end{eqnarray*}
  where $b_{jl}=a_i$ if and only if $T(i)=(j,l)$. Therefore, $\I$ is an information set for $\C$ if and only if $T^{-1}(\I)$ is an information set for $\C^*$. Usually, we omit the reference to the original cyclic code and the isomorphism $T$ in the notation of the new abelian code, and we will write $\C$ instead of $\C_{(\C^*,T)}$; those references will by clear by the context.
  
  For the rest of this section we assume that we have fixed a choice of $\alpha$ and $T$ (consequently, the roots $\beta_1,\beta_2$ are also fixed).

Then, the goal of this section is to describe the values of the parameters defined in (\ref{parametrosm}) and (\ref{M(e)}), used to get a set of check positions  for $\C$ (see \ref{checkpositions}), just in terms of the defining set of $\C^*$. This will allow us to define an information set for the original cyclic code $\C^*$ without any mention to the abelian code $\C$. 

%\begin{remark}
% It is well known that any $k$ consecutive positions form an information set for a cyclic code of dimension $k$ (see, for example, \cite{HP}). We refer to this set as the classical information set for a cyclic code.
%\end{remark}

\begin{remark}\label{muchos conjuntos info}

%Given a cyclic code $\C^*$, to define the corresponding abelian code $\C=\C_{(\C^*,T)}$ we need to fix a primitive $n$-th root of unity $\alpha$ and an isomorphism $T:\Z_n\rightarrow\Z_{r_1}\times\Z_{r_2}$. On the one hand, from $\alpha$ we determine the defining set of $\C^*$ with respect to it, denoted by $\D^*$. On the other hand, we use $T$ to construct $\D=T(\D^*)$ the defining set for $\C$ with respect to a suitable set of roots of unity $\beta_1,\beta_2$ (see paragraph after Lemma~\ref{T-orbitas}). Then $T^{-1}(\Gamma(\C))$ is an information set for $\C^*$.

Let us observe that, since the defining set $\D^*=D_\alpha(\C^*)$ depends on the choice of $\alpha$, if we fix another one $\beta$, we get a new defining set $D_\beta(\C^*)$ and consequently we obtain a different abelian code $\C$. However, this new abelian code has a defining set with the same structure of $q$-orbits so it yields the same set $\Gamma(\C)$ (see paragraph after Theorem~\ref{teorema principal check positions}).

On the other hand, it is easy to prove that any change of the isomorphism $T$  is equivalent to a change of the primitive root $\alpha$ in order to get the same abelian code in $\A(r_1,r_2)$. Nevertheless, this implies that, since we get the same set $\Gamma(\C)$, we could obtain a different set of check positions $T^{-1}(\Gamma(\C))$ at the end. Therefore, by the method we are describing in this section we get at most as many information sets as there are isomorphisms from $\Z_n$ to $\Z_{r_1}\times\Z_{r_2}$ exist.

\end{remark}

Let $\overline {\D^*} \subseteq \D^*\subseteq \Z_n$ be a complete set of representatives of the $2$-cyclotomic cosets modulo $n$ in $\D^*$. As we have noted, for any $e\in \Db$, $T(C_n(e))=Q(T(e))$, so $T(\overline{D^*})$ is a complete set of representatives of the $2$-orbits modulo $(r_1,r_2)$ in $T(D^*)$. However, it might not be a set of restricted representatives according to Definition~\ref{restricted representatives}. Then we introduce the following definition.

\begin{definition}\label{suitable}
 Let $\overline {\D^*} \subseteq \D^*$ be a complete set of representatives of the $2$-cyclotomic cosets modulo $n$ in $\D^*$. Then $\overline{\D^*}$ is said to be a suitable set of representatives if $T(\overline{\D^*})$ is a set of restricted representatives of the $2$-orbits in $T(\D^*)$.
\end{definition}

The next lemma shows that we will always be able to take a suitable set of representatives in $\D^*$.

\begin{lemma}
 Let $\D^*$ be the defining set of a cyclic code $\C^*\subseteq\A(n)$. Let $\C\subseteq\A(r_1,r_2)$ be the abelian code with defining set $T(\D^*)$. Then, there always exists a suitable set of representatives $\overline{\D^*}\subseteq \D^*$.
\end{lemma}
\begin{proof}
From Lemma~\ref{T-orbitas}, for any $e\in \D^*$ we have that $Q(T(e))_1=C_{r_1}(T_1(e))$. Let $\Db$ be a complete set of representatives of the $2$-cycloctomic cosets modulo $n$ in $\D^*$. Take $A$ a complete set of the $2$-cyclotomic cosets modulo $r_1$ in $\{T_1(e): e\in \D^*\}$. Now, take $e\in\Db$ and let $a\in A$ be such that $T_1(e)\in C_{r_1}(a)$. Then there exists $e'\in C_n(e)$ that satisfies $T_1(e')=a$; if $e'\neq e$ then we redefine $\Db$ by replacing $e$ with $e'$. This concludes the proof.
\end{proof}
%\vspace{.5cm}

%\begin{definition}\label{U}
%Let $D^*$ be a defining set of a cyclic code $\C^*\subseteq\A(n)$ and $\overline {\D^*}$ a suitable set of %representatives. 

\begin{example}\label{examplesuitable}
 Let us consider binary cyclic codes of length $n=21$, so $r_1=3, r_2=7$. Let $\C^*\subseteq\A(21)$ be the cyclic code such that its defining set, with respect to a $21$-th primitive root of unity $\alpha$, is the following union of $2$-cyclotomic cosets modulo $21$, $\D^*=\{1,2,4,8,11,16\}\cup \{3,6,12\}\cup \{7,14\}$. Take $T:\Z_{21}\rightarrow\Z_{3}\times\Z_{7}$ the isomorphism given by the Chinese Remainder Theorem. Then $\C$ denotes the abelian code in $\A(3,7)$ with defining set $\D=T(\D^*)$, that is, the following union of $2$-orbits modulo $(3,7)$, $T(\D)=\{(1,1),(2,2),(1,4),(2,1),(2,4),(1,2)\}\cup\{(0,3),(0,6),(0,5)\}\cup\{(1,0),(2,0)\}$. The set $B=\{2,3,7\}\subseteq\D^*$ is a complete set of representatives of the $2$-cyclotomic cosets in $\D^*$; however it is not a suitable set of representatives since $T(B)=\{(2,2),(0,3),(1,0)\}$ is not a restricted set of representatives in $\D$ (note that $C_3(1)=C_3(2)$). We may solve this problem by replacing $2$ by $1$. Then $\Db=\{1,3,7\}$ is a suitable set of representatives because $T(\Db)=\{(1,1),(0,3),(1,0)\}$ is a restricted set of representatives in $\D$.

% $T_1(B)=\{0,1,2\}$ is not a complete set of representatives of the $2$-cyclotomic cosets modulo $3$ .
\end{example}

Now we deal with the construction of a set of check positions for the abelian code $\C=\C_{(\C^*,T)}\subseteq \A(r_1,r_2)$ just in terms of $\D^*$, the defining set of the cyclic code $\C^*$.

Given $\overline {\D^*}\subseteq \D^*$ a suitable set of representatives, we consider $\sim$ the equivalence relation on $\overline {\D^*}$ given by the rule 
\begin{equation}\label{relacionU}
 a\sim b\in \overline {\D^*} \text{ if and only if } a\equiv b \mod r_1.
\end{equation}
  From now on we denote by $\Db$ a suitable set of representatives and $\U \subseteq \Db$ a complete set of representatives of the equivalence classes related to $\sim$. In addition, for any $u\in \U$ we write 
\begin{equation*}
\Or(u)=\{a\in\Db \mid a\sim u \}.
\end{equation*}

Observe that if $\overline \D=T(\overline{\D^*})$ then ${\overline D_1}=T_1(\U)$. Furthermore, for any $e\in \overline{\D^*}$ there exists a unique $u\in \U$ such that $T_1(u)=T_1(e)$. By abuse of notation, we will write 
$$C_{r_1}(e)=C_{r_1}(T_1(e))=C_{r_1}(T_1(u))=C_{r_1}(u).$$

\begin{example}
 Following Example \ref{examplesuitable}, from the suitable set of representatives $\Db=\{1,3,7\}$ we may define, for instance, $\U=\{1,3\}$. Note that $\Or(1)=\{1,7\}$. 
\end{example} 
 
 \begin{lemma}\label{relacion mTe con su i}
 For any $e\in \overline{\D^*}$ one has that $$m\left(T(e)\right)=\frac{|C_n(e)|}{|C_{r_1}(u)|},$$ where $u$ is the unique element in $\U$ such that $T_1(u)=T_1(e)$.
\end{lemma}
\begin{proof}
 By using the equalities in Lemma~\ref{T-orbitas} and the definition of $T$ one has that $\left|Q(T(e))\right|=m(T_1(e))\cdot m(T(e))$. On the other hand $m(T_1(e))=m(T_1(u))=\left|C_{r_1}(T_1(u))\right|=\left|C_{r_1}(u)\right|$, so we are done.
\end{proof}
 
 The following proposition shows how we can write the parameters (\ref{parametrosm}) over the elements of $\overline{\D}=T(\Db)\subseteq\Z_{r_1}\times\Z_{r_2}$ in terms of elements in $\U\subseteq \Db$.
 
 \begin{proposition}\label{relacion del m de los de Re con su u}
  Take $e\in \Db$ and $u\in \U$ (unique) such that $T_1(u)=T_1(e)$. Let $\overline \D=T(\Db)\subseteq\Z_{r_1}\times\Z_{r_2}$. For each $k\in R(T_1(e))$ there exists a unique $v_k\in  \Or(u)\subseteq\Db$ satisfying $T(v_k)=(T_1(e),k)$ and 
  $$m(T_1(e),k)=\frac{|C_n(v_k)|}{|C_{r_1}(u)|}.$$
\end{proposition}
\begin{proof}
Take $k\in R(T_1(e))$. Then $(T_1(e),k)\in\overline\D=T(\Db)$, so there exists a unique $v_k\in\Db$ such that $T(v_k)=(T_1(e),k)$. Moreover, $v_k\in \Or(u)$ because $T_1(v_k)=T_1(e)=T_1(u)$. Now
$$\left|C_n(v_k)\right|=\left|Q(T(v_k))\right|=\left|Q(T_1(e),k)\right|=m(T_1(e))\cdot m(T_1(e),k)=m(T_1(u))\cdot m(T_1(e),k)=\left|C_{r_1}(u)\right|\cdot m(T(e_1),k).$$

\end{proof}
 
Now, we use the previous result to obtain the values $M(\cdot)$ (defined in (\ref{M(e)})) corresponding with the set $\overline{\D}=T(\Db)\subset\Z_{r_1}\times\Z_{r_2}$, in terms of the elements in $\U\subseteq\Z_n$.

\begin{theorem}\label{calculo del Mi}
%Let $\C^*\subseteq \A(n)$ be a cyclic code with defining set $\D^*\subseteq \Z_n$, where $n=r_1r_2$ and $(r_1,r_2)=1$.  Let $\C\subseteq \A(r_1,r_2)$ be the abelian code with defining set $\D= T(\D^*)$. Take $\Db$ and $\U$ as above. Then, for
For each $(e_1,e_2)\in \overline\D=T(\Db)$, there exists a unique $u\in \U$ such that $T_1(u)=e_1$ and 
  $$M(e_1)=\frac{1}{|C_{r_1}(u)|}\sum_{v\in  \Or(u)}|C_n(v)|=\sum_{v\in \Or(u)}|C_{2^{|C_{r_1}(u)|},r_2}(T_2(v))|.$$
\end{theorem}

\begin{proof}
Take $(e_1,e_2)\in \overline\D$ and $e\in \Db$ such that $T(e)=(e_1,e_2)$. Then
$$M(e_1)=\sum\limits_{k\in R(e_1)}m(e_1,k)=\sum\limits_{k\in R(e_1)}m(T_1(e),k)=M(T_1(e)).$$

By Proposition~\ref{relacion del m de los de Re con su u}, for any $k\in R(e_1)$ there exists a unique $v_k\in \Or(u)\subseteq\Db$ such that $T(v_k)=(T_1(e),k)$ and $m(T_1(e),k)=\frac{|C_n(v_k)|}{|C_{r_1}(u)|}$. Now, note that if $k\neq k'\in R(e_1)$ then $v_k\neq v_{k'}$ because $T(v_k)=(e_1,k)\neq(e_1,k')=T(v_{k'})$. Finally, if $v\in \Or(u)$ then $T_1(v)=T_1(u)=e_1$, and then there exists $k\in R(e_1)$ such that $T(v)=(e_1,k)$, so $v=v_k$. This finishes the proof.
\end{proof}

\vspace{0.5cm}

To sum up, fixing a $n$-th primitive root of unity $\alpha$ and an isomorphim $T:\Z_n\rightarrow \Z_{r_1}\times\Z_{r_2}$, let $\C^*$ be a cyclic code in $\A(n)$ with defining set with respect to $\alpha$, $\D^*$, and let $\C=\C_{(\C^*,T)}\subseteq \A(r_1,r_2)$ be the abelian code with defining set $\D=T(\D^*)$ (taking as reference the suitable primitive roots of unity mentioned at the beginning of this section). We have shown that we can take a suitable set of representatives $\Db\subseteq\D^*$ and a set $\U\subseteq \overline {\D^*}$, with $T_1(\U)=\overline\D_1$ $(\overline \D=T(\Db))$, in such a way that we are able to construct the set of check positions given in (\ref{checkpositions}) for $\C$ in terms of $\U$ as follows: 

 Since for any $e_1\in \overline\D_1$  there exists a unique element $u\in \U$ that satisfies $T_1(u)=e_1$, by abuse of notation, we may write $M(u)=M(T_1(u))=M(e_1)$. Then, the set of values (\ref{M(e)}) can be described as

\begin{equation}\label{M(u)}
 \left\{M(u)=\frac{1}{|C_{r_1}(u)|}\sum_{v\in \Or(u)}|C_n(v)| \tq u\in \U\right\}
\end{equation} 
which yields the sequence $f_{1}>\dots>f_{s}>0=f_{s+1}$ (see (\ref{sequence f[i]})). On the other hand, for any $k=1,\dots,s$ the values (\ref{lasg}) can be computed as 

$$g_k=\sum_{\doble{u\in \U}{M(u)\geq f_k}}|C_{r_1}(u)|$$
which give us the sequence $g_1<g_2<\cdots<g_s$ (see (\ref{sequence gi})) and hence the set $\Gamma(\C)$.

\begin{remark}
 From the development of the method we have described, the reader may note that the isomorphism $T$ has no influence on the construction of $\Gamma(\C)$. So, as we have said in Remark~\ref{muchos conjuntos info}, we obtain at most as many information sets as there are isomorphisms. 
\end{remark}

\begin{example}
 We apply the results in this section to the code $\C^*$ given in Example~\ref{examplesuitable}. Recall that its defining set is $\D^*=\{1,2,4,8,11,16\}\cup \{3,6,12\}\cup \{7,14\}$ and we were using the isomorphism given by the Chinese Remainder Theorem. We have chosen $\Db=\{1,3,7\}$ as suitable set of representatives and $\U=\{1,3\}$. Then, from (\ref{M(u)}) we have 
  \begin{eqnarray*}
   M(1)&=&\frac{1}{\left|C_3(1)\right|}\left(\left|C_{21}(1)\right|+\left|C_{21}(7)\right|\right)=\frac{1}{2}(6+2)=4,\\
   M(3)&=&\frac{1}{\left|C_3(3)\right|}\cdot \left|C_{21}(3)\right|=\frac{1}{1}\cdot 3=3,
  \end{eqnarray*}
 and so $f_1=4>f_2=3>f_3=0$. On the other hand, $g_1=m(1)=2<g_2=2+m(0)=3$.  These sequences yield the set of check positions for $\C$
 $$\Gamma(\C)=\{(0,0),(0,1),(0,2),(0,3),(1,0),(1,1),(1,2),(1,3),(2,0),(2,1),(2,2)\}.$$
Finally, we have that $T^{-1}(\Gamma(\C))=\{0,1,2,3,7,8,9,10,14,15,16\}$ is a set of check positions for $\C^*$.
 
 Now, let us consider the isomorphism $\widehat T:\Z_{21}\rightarrow \Z_3\times\Z_7$ given by $\widehat T(1)=(1,2)$. Then 
 $${\widehat T}^{-1}(\Gamma(\C))=\{0,1,4,7,8,11,12,14,15,18,19\},$$
which is a different set of check positions for $\C^*$. 
 
\end{example}

%\begin{remark}
% All this section can be developed analogously under the more general assumption $n=r_1 \cdots r_l$ where $l\geq 2$. However the results become more complex, so for the convenience of the reader we have just include the case $l=2$.
%\end{remark} 
\section{Reed-Muller codes}\label{RM}

In this section we shall introduce Reed-Muller codes from the group-algebra point of view. 
%Indeed this could be considered as the natural approach for generalized Reed-Muller codes. However, as reiterated all throughout this document, we are interested just in the binary case. 
Specifically, we are going to present Reed-Muller codes as a type of code contained in the family of so-called affine-invariant extended cyclic codes (see, for instance, \cite{Hb},\cite{Hb2} or \cite{Charpin}).

Recall that $\F$ denotes the binary field. Let $G$ be the additive subgroup of the field of $2^m$ elements. So $G$ is an elementary abelian group of order $\left|G\right|=2^m$ and $G^*=G\setminus\{0\}$ is a cyclic group. From now on we write $n=2^m-1$. We consider the group algebra $\F G$ which will be the ambient space for Reed-Muller codes. We denote the elements in $\F G$ as $\sum_{g\in G} a_g X^g$ and then the operations are written as follows
 \begin{eqnarray*}
  \sum\limits_{g\in G}a_g X^g+\sum\limits_{g\in G}b_g X^g&=&\sum\limits_{g\in G}(a_g+b_g) X^g\\
  c\cdot\sum\limits_{g\in G}a_g X^g&=&\sum\limits_{g\in G}(c\cdot a_g) X^g\ (c\in \F)\\
  \left(\sum\limits_{g\in G}a_g X^g\right)\cdot\left(\sum\limits_{h\in G}b_h X^h\right)&=&\sum\limits_{g\in G}\left(\sum\limits_{g_1+g_2=g} a_{g_1} b_{g_2} \right)X^g.
 \end{eqnarray*}

Notice that $X^0$ is the unit element in $\F G$. The following definitions introduce the family of affine-invariant extended cyclic codes. All throughout this section we fix  $\alpha$, a generator of the cyclic group $G^*$, that is, a primitive $n$-th root of unity. %$(2^m-1)$-th root of unity.

\begin{definition}
 Let $\alpha$ be a generator of $G^*$. A code $\C\subseteq \F G$ is an extended cyclic code if for any $\sum\limits_{g\in G}a_g X^g\in \C$ one has that $\sum\limits_{g\in G}a_g X^{\alpha g}\in \C$ and $\sum\limits_{g\in G}a_g=0$.
\end{definition}

\begin{definition}
 We say that an extended cyclic code $\C\subseteq \F G$ is affine-invariant if for any $\sum\limits_{g\in G}a_g X^g\in \C$ one has that $\sum\limits_{g\in G}a_g X^{hg+k}$ belongs to $\C$ for all $h,k\in G, h\neq 0$.
\end{definition}

It is clear that if $\C\subseteq \F G$ is affine-invariant then $\C$ is an ideal in $\F G$ and $\C^*\subseteq \F G^*$, the punctured code at the position $X^0$, is cyclic in the sense that it is the projection to $\F G^*$ of the image of a cyclic code via the map 
 \begin{eqnarray}\label{inmersionciclicos}
  \nonumber \A(n) &\longrightarrow& \F G\\
   \sum\limits_{i=0}^{n-1}a_{i}X^i&\hookrightarrow& \left(-\sum\limits_{i=0}^{n-1} a_i\right)X^0+\sum\limits_{i=0}^{n-1}a_i X^{\alpha^i},
  \end{eqnarray}
where $\alpha$ is the fixed $n$-th root of unity.

Now, for any $s\in\{0,\dots,n=2^m-1\}$ we consider the $\F$-linear map $\phi_s:\F G\rightarrow G$ given by
\begin{equation*}
 \phi_s\left(\sum\limits_{g\in G}a_g X^g\right)=\sum\limits_{g\in G}a_g g^s
\end{equation*}
where we assume $0^0=1\in\F$ by convention.

\begin{definition}
 Let $\C\subseteq\F G$ be an affine-invariant code. The set 
 $$D(\C)=\{i\mid \phi_i(x)=0 \text{ for all } x\in \C\}$$
 is called the defining set of $\C$.
\end{definition}

Note first that since $\C$ is an extended-cyclic code, one has that $0\in D(\C)$ because $\phi_0\left(\sum\limits_{g\in G}a_g X^g\right)=\sum\limits_{g\in G}a_g g^0=\sum\limits_{g\in G}a_g$. Furthermore, it follows from the equality $\phi_{2s}(x)=(\phi_s(x))^2\ (x\in \C)$ that $D(\C)$ is a union of $2$-cyclotomic cosets modulo $n$. On the other hand, keeping in mind the map (\ref{inmersionciclicos}), we talk about the set of zeros an the defining set of $\C^*$ when we are making reference to those of the corresponding cyclic code in $\A(n)$. So, fixing $\alpha$ a $n$-th primitive root of unity, the zeros of the cyclic code $\C^*$ are $\{\alpha^s\mid s\in D(\C), s\neq 0\}$ and the defining set of $\C^*$ (according to Definition~\ref{definingset} and with respect to that root of unity), is $D(\C^*)=D(\C)\setminus\{0\}$.

It is easy to prove that any affine-invariant code is totally determined by its defining set. Conversely, any subset of $\{0,\dots,n\}$ which is a union of $2$-cyclotomic cosets and contains $0$ defines an affine-invariant code in $\F G$.

\begin{remark}
 It may occur that $n,0\in D(\C)$ which could yield confusion considering the $2$-cyclotomic cosets modulo $n$. Those elements are considered distinct and they indicate different properties of the code $\C$; namely, $0$ always belongs to $D(\C)$ because $\C$ is an extended cyclic code, while if $n$ belongs to $D(\C)$ then the cyclic code $\C^*$ is even-like which implies that $\C$ is a trivial extension.
\end{remark}

Finally, to introduce the family of Reed-Muller codes as affine-invariant codes we need to recall the notions of binary expansion and $2$-weight. For any natural number $k$ its binary expansion is the sum $\sum_{r\geq 0} k_r 2^r=k$ with $k_r\in \{0,1\}$. The $2$-weight or simply weight of $k$ is ${\rm wt}(k)=\sum_{r\geq 0}k_r$.

\begin{definition}\label{defRM}
 Let $0< \rho< m$. The Reed-Muller code of order $\rho$ and length $2^m$, denoted by $R(\rho,m)$, is the affine-invariant code in $\F G$ with defining set 
 $$D(R(\rho,m))=\{i\mid 0\leq i<2^m-1 \text{ \rm{and} } {\rm wt}(i)< m-\rho\}.$$
\end{definition}

From the classical point of view, Reed-Muller codes are also defined for the cases $\rho=0,m$, but they correspond with the trivial cases $R(m,m)=\F G$ and $R(0,m)=\langle\sum_{g\in G} g\rangle$ (the repetition code) which do not have interest in the context of this paper. The following result summarizes some well known results about Reed-Muller codes that we will refer to further. The reader may see \cite{Hb2}.

\begin{proposition}\label{basicsRM}
  Let $0< \rho< m$.
 \begin{enumerate}
  \item The Reed-Muller code $R(\rho,m)$ is a code of length $2^m$, dimension   
          $$k= {m\choose 0}+{m\choose 1}+\cdots+{m\choose \rho}$$
         and minimum distance $2^{m-\rho}$. 
  %\item $R(m,m)=\F G$ and $R(0,m)=\langle\sum\limits_{g\in G} g\rangle$, that is, $R(0,m)$ is the repetition code.
  \item $R(m-1,m)=\{\sum\limits_{g\in G}a_g g\mid \sum\limits_{g\in G}a_g =0\}$, that is, the code of all even weight vectors in $\F G$.
  \item $R(\rho,m)^\bot=R(m-\rho-1,m)$.
 \end{enumerate}
\end{proposition}

Since $R(m-1,m)^\bot=R(0,m)$ the problem of searching for information sets has no interest in the case $\rho=m-1$, so in the rest of the paper we will assume that $\rho< m-1$.

As a consequence of Definition~\ref{defRM} we have that the defining set of the punctured code $R^*(\rho,m)$, at the position $X^0$ and with respect to the fixed $n$-th root of unity $\alpha$, is given by the following union of cyclotomic cosets modulo $n$
    $$D(R^*(\rho,m))=\bigcup_{i\in D(R(\rho,m))\setminus\{0\}}C_n(i).$$
    
In the following sections we will deal with the application of the results contained in Section~\ref{cyclics} to the cyclic code $R^*(\rho,m)$ in order to obtain an information set for $R(\rho,m)$. To use a notation congruent with that used in that section we write $\D^*=D(R^*(\rho,m))$. We assume that there exist integers $r_1,r_2$ such that $n=2^m-1=r_1\cdot r_2$ with $r_1, r_2$ odd and $\gcd(r_1,r_2)=1$. Now, we fix some notation and introduce some definitions and important results that will be needed.

\begin{remark}\label{check positions union 0}

Note that for any fixed primitive root of unity $\alpha$, following the notation used to describe the elements in $\F G$, $\left(-\sum\limits_{i=0}^{n-1} a_i\right)X^0+\sum\limits_{i=0}^{n-1}a_i X^{\alpha^i}\in\F G$, a set $\I\subseteq\{0,\alpha^0,\dots,\alpha^{n-1}\}$  is an information set for a code $\C\subseteq\F G$ with dimension $k$, if $\left|\I \right|=k$ and $\C_\I=\F^{\left|\I \right|}$.
 
 On the other hand, since $R^*(\rho,m)$ is contained in $\F G^*$ we have that any information set for it will be a subset of $\{\alpha^0,\dots,\alpha^{n-1}\}$. Obviously, an information set for $R^*(\rho,m)$ is an information set for $R(\rho,m)$ too. However, note that if $\Gamma$ is a set of check positions for $R^*(\rho,m)$ then a set of check positions for $R(\rho,m)$ is $\Gamma\cup\{0\}$.
\end{remark}

Now, we need to fix some notation. For any integer $0<K<m-1$ we define
\begin{eqnarray*}
 \Omega(K)&=&\left\{0< j<  2^m-1\tq {\rm wt}(j)=K\right\}
 =\left\{2^{t_1}+\dots+2^{t_{K}}\tq 0\leq t_1<\dots<t_{K}<m\right\}
\end{eqnarray*}

and hence, $$\D^*=\bigcup_{K=1}^{m-\rho-1}\{0< j< 2^m-1 \tq {\rm wt}(j)=K\}=\bigcup_{K=1}^{m-\rho-1}\Omega(K).$$

\begin{example}
 Take $m=4$. Then $n=2^4-1=15, r_1=3, r_2=5$. For these parameters one has that
   $$\Omega(1)=\{1,2,4,8\}, \Omega(2)=\{3,5,6,9,10,12\} \text{ and } \Omega(3)=\{7,11,13,14\}.$$
 
 The code $R(1,4)$ has defining set $\D(R(1,4))=\{0,1,2,3,4,5,6,8,9,10,12\}$ and so $\D^*=\D(R^*(1,4))=\Omega(1)\cup \Omega(2)=\{1,2,3,4,5,6,8,9,10,12\}$. 
 
 Finally, the code $R(2,4)$ has defining set $\D(R(2,4))=\{0,1,2,4,8\}$, which implies $\D^*=\D(R^*(2,4))=\Omega(1)$.
\end{example}

Now, we take $\Db$ a suitable set of representatives in $\D^*$ and a set $\U$ as in the previous section (see Definition~\ref{suitable} and subsequent paragraphs).

We are interested in handling convenient elements in the fixed suitable set of representatives in order to compute the necessary cyclotomic cosets. We will describe them in Theorem~\ref{teoremaacotados}. First, we need a lemma.

\begin{lemma}\label{desigualdadestope}
 Let $j<i<K<m$ be natural numbers. Then
 \begin{enumerate}
  \item $m-\topeK{im}= \topeK{(K-i)m}$ or $\topeK{(K-i)m}+1$.
  \item $\topeK{im}-\topeK{jm} = \topeK{(i-j)m}$ or $\topeK{(i-j)m}+1$.
 \end{enumerate}
\end{lemma}
\begin{proof}
 Let $\topeK{im}=\frac{im}{K}-\delta$ with $0\leq \delta <1$. Then 
 $$m-\topeK{im}=m-\frac{im}{K}+\delta=\frac{(K-i)m}{K}+\delta.$$
So, if $\delta=0$ then $\topeK{(K-i)m}=m-\topeK{im}$, and if $\delta>0$ then $\topeK{(K-i)m}=m-\topeK{im}-1$. This proves {\textit 1)}.

Now, let $\topeK{jm}=\frac{jm}{K}-\delta'$ with $0\leq \delta'<1$. Then 
$$\topeK{im}-\topeK{jm}=\frac{(i-j)m}{K}+\delta'-\delta.$$
Note that $-1<\delta'-\delta<1$. So, if $\delta\geq\delta'$ then $\topeK{(i-j)m}=\topeK{im}-\topeK{jm}$; otherwise, $\topeK{(i-j)m}=\topeK{im}-\topeK{im}-1$. This finishes the proof.
\end{proof}

\vspace{.5cm}

 It is clear that if $e=2^{s}\in \Omega(1)$, with $0\leq s <m$, then $C_n(e)=C_n(1)$. The following theorem establishes a more general result for $\Omega(K)$ with $1<K<m-1$.

\begin{theorem}\label{teoremaacotados}
Let $e=2^{s_0}+2^{s_1}+\dots+2^{s_{K-1}}\in \Omega(K)$ with $0\leq s_0<\dots <s_{K-1}<m$ and $1<K<m-1$. Then there exists  
\begin{equation}\label{acotados} e'= 1+2^{t_1}+\dots +2^{t_{K-1}} ,\end{equation}
with $t_i\leq \topeK{im}$ for all $i=1,\dots,K-1$, such that $C_n(e)=C_n(e')$.
\end{theorem}

\begin{proof}
We can assume w.l.o.g. that $e=1+2^{s_1}+2^{s_2}+\dots+2^{s_{K-1}}\in \Omega(K)$ with $0< s_1<\dots <s_{K-1}<m$ and $1<K<m-1$. If $e$ verifies the required conditions we are done, so suppose that $e$ do not satisfies them and let $\delta_1, \delta_2$ be the integers such that 
  $$\begin{array}{ccl}  
  s_i\leq \topeK{im} &\text{ for all } &i=1,\dots,\delta_1-1,\\
  s_{\delta_1}>\topeK{\delta_1m},&&\\
  s_{K-i}\leq\topeK{(K-i)m} &\text{ for all }& i=1,\dots,\delta_2;   
  \end{array}$$
and $K>\delta_1+\delta_2$, that is, the binary expansion of $e$ has at least $\delta_1+\delta_2-1$ exponents satisfying the desired condition. Observe that $\delta_1\geq 1, \delta_2\geq 0$.
  
We define $u=m-s_{\delta_1}$ and consider $e'=2^ue$ modulo $n$. Then, it is easy to see that  $e'=1+2^{s'_1}+\dots+2^{s'_{K-1}}$, with $0<s'_1<\dots <s'_{K-1}<m$, where
 \begin{equation}\label{sprima1}  s'_j\equiv s_{\delta_1+j}+u\  (\text{ mod } m)\ \text{ if }\ j=1,\dots, K-\delta_1-1\end{equation}   
\begin{equation*}\label{sprima2}  s'_{K-\delta_1}=u \end{equation*}
 \begin{equation}\label{sprima3} s'_j=s_{\delta_1+j-K}+u\  \text { if }\ j=K-\delta_1+1,\dots,K-1.\end{equation}

We claim that $s'_j\leq\topeK{jm}$ for any $j=K-(\delta_1+\delta_2),\dots,K-1$. Let us define $A=\{K-\delta_1+1,K-\delta_1+2,\dots,K-1\}$ and $B=\{K-(\delta_1+\delta_2),K-(\delta_1+\delta_2)-1,\dots,K-\delta_1-1\}$, then $A\cup B\cup\{K-\delta_1\}=\{K-(\delta_1+\delta_2),\dots,K-1\}$. Note that in the cases $\delta_1=1$, $\delta_2=0$ one has that $A=\emptyset$, $B=\emptyset$ respectively.

On  the one hand, by (\ref{sprima3}), one has that $s'_j=s_{\delta_1+j-K}+u$ for any $j\in A$. Let us observe that $s_{\delta_1}-s_{\delta_1-i}>\topeK{\delta_1 m }-\topeK{(\delta_1-i)m}\geq \topeK{im}$ for any $i=1\dots,\delta_1-1$ (see Lemma~\ref{desigualdadestope}). Then, $$s'_j=m-s_{\delta_1}+s_{\delta_1-(K-j)}=m-(s_{\delta_1}-s_{\delta_1-(K-j)})<m-\topeK{(K-j)m}\leq \topeK{jm}+1$$
(see Lemma~\ref{desigualdadestope}), and so $s'_j\leq \topeK{jm}$ for any $j\in A$. Furthermore, 
$$s'_{K-\delta_1}=u=m-s_{\delta_1}<m-\topeK{\delta_1 m}\leq\topeK{(K-\delta_1)m}+1,$$
which implies $s'_{K-\delta_!}\leq\topeK{(K-\delta_1)m}$.

On the other hand, by (\ref{sprima1}), for any $j\in B$ we have that $s'_j\equiv s_{\delta_1+j}+ u =m -s_{\delta_1}+s_{\delta_1+j}=m-(s_{\delta_1}-s_{\delta_1+j})$ mod $m$, so $s'_j=s_{\delta_1+j}-s_{\delta_1}$ (note that $s_{\delta_1+j}>s_{\delta_1}$). This leads us to 
$$s'_j=s_{\delta_1+j}-s_{\delta_1}<\topeK{(\delta_1+j)m}-\topeK{\delta_1m}\leq \topeK{jm}+1,$$
and then $s'_j\leq\topeK{jm}$, so we are done.

Therefore, $e'=2^u e$ satisfies that its binary expansion has at least $\delta_1+\delta_2$ exponents verifying the desired condition, one more than in the case of $e$. If we repeat the argument successively by replacing $e$ by $e'$ then we get what we wanted. This finishes the proof.
\end{proof}

The previous result shows that for any element  $e\in \Omega(K), 1 < K < m-1,$ there exists another element $e'\in C_n(e)$ satisfying condition (\ref{acotados}). As we will see in the next section, in the case $K=2$, this fact becomes essential for our purposes. In addition, let us observe that the elements in a given suitable set of representatives might not satisfy condition (\ref{acotados}); we include the next example to show that case.

\begin{example}\label{suitablenot10}

Let us consider the Reed-Muller code $R(3,6)$. So $m=6, n=2^6-1=63, r_1=7, r_2=9$. For these parameters we have $$\Omega(1)=C_{63}(1)=\{1,2,4,8,16,32\}$$ 
and $$\Omega(2)=C_{63}(3)\cup C_{63}(5)\cup C_{63}(9)=\{3,5,6,9,10,12,17,18,20,24,33,34,36,40,48\}.$$
The defining set of $R^*(3,6)$ is $\D^*=\Omega(1)\cup\Omega(2)$. The set $\Db=\{1,3,10,36\}$ is a suitable set of representatives according to Definition~\ref{suitable}; however while $1$ and $3$  satisfy condition (\ref{acotados}) the elements, $10=2+2^3$ and $36=2^2+2^5$ do not satisfy it. The reader may check that there does not exist any suitable set of representatives such that all its elements satisfy condition (\ref{acotados}).
%As the reader will see, we shall always assume that $1$ belongs to the set of suitable representatives taken; under this assumption the set $\Db=\{1,3,5,36\}$ is the unique suitable set of representatives.
  
\end{example}

%\begin{theorem}\label{Conjunto Omega(K)}
% Let $\D^*=D(R^*(\rho,m))$ be the defining set of the punctured cyclic code of the Reed-Muller code of order $\rho$ and length $2^m-1$. Then we can choose a suitable set of representatives $\Db=\bigcup_{K=1}^{m-\rho-1} \widehat{\Omega}(K)$ satisfying that  
% \[\widehat{\Omega}(K)=\left\{1+2^{t_1}+\dots+2^{t_{K-1}}\tq 0\leq t_1< t_2<\dots < t_r\leq \topeK{rm},\;r=1,\dots K-1\right\}\]
%for any integer $1< K\leq m-\rho$ and $\widehat{\Omega}(1)=\{1\}$. 
% 
%\end{theorem}

%\begin{example}
%Por ejemplo, sean $K=10$, $r=4$ y $t=7$. Entonces $s_4>\topeK{4m}$ quedan descontrolados $s_4$, $s_5$ y $s_6$. Hay 6 elementos no cero que cumplen la propiedad. Entonces $$s_4\mapsto 0,\;s_3\mapsto s_9,\dots,0\mapsto s_6,\;s_9\mapsto s_5\dots,s_7\mapsto s_3,$$ como debe ser, pues $7-4=3$. Ahora la propiedad la cumple de $s_3$ en adelante, o sea, 7 elementos.
%\end{example}

\section{Information sets for first-order Reed-Muller codes}\label{first-order}

In this section we are applying the results of Section~\ref{cyclics} to the punctured codes of first-order Reed-Muller codes. Specifically, we shall construct an information set for the punctured code $R^*(1,m)$ by using those techniques and then we shall give an information set for the Reed-Muller code $R(1,m)$. To use the mentioned results we need to assume that there exist odd integers $r_1,r_2$ such that $n=2^m-1=r_1\cdot r_2$ and $\gcd(r_1,r_2)=1, r_1,r_2>1$.

All throughout this section we fix a primitive $n$-th root of unity $\alpha$ and $T:\Z_n\rightarrow \Z_{r_1}\times\Z_{r_2}$ an arbitrary isomorphism of groups.

 We have two different possibilities to get an information set for the code $R(1,m)$; namely, from an information set of $R^*(1,m)$, which comes from the defining set of $R^*(1,m)$, and from a check of set positions of $R(1,m)^\bot=R(m-2,m)$, which depends on the defining set of $R^*(m-2,m)$. In general it is more convenient to develop the results in terms of $R^*(m-2,m)$ whose defining set is much smaller than that of $R^*(1,m)$.
 
 Let us denote $\C^*=R^*(m-2,m)$. Then, following the notation fixed in the previous section, we have that 
     $$\D^*=\Omega(1)$$
where $\Omega(1)=\{2^t \tq 0\leq t<m\}$. Observe that $\Omega(1)=C_n(1)$ so we take $\Db=\{1\}$ as our suitable set of representatives (see Definition~\ref{suitable}). Then $\U=\{1\}$. % (see (\ref{relacionU})). 

 The following theorem gives us the description of an information set for the code $R(1,m)$. We denote by ${\rm Ord}_{r_1}(2)$ the order of $2$ modulo $r_1$, that is, the smallest integer such that $2^{{\rm Ord}_{r_1}(2)}\equiv 1$ modulo $r_1$.
 
\begin{theorem}\label{teoremainfosetprimerorden}
  Suppose that $n=2^m-1=r_1\cdot r_2$, where $r_1,r_2$ are odd integers such that $\gcd(r_1,r_2)=1, r_1,r_2>1$. Let $\C^*=R^*(m-2,m)$ and $a={\rm Ord}_{r_1}(2)$. Let $\D^*\subseteq \Z_n$ be the defining set of $\C^*$ with respect to $\alpha$, a primitive $n$-th root of unity, and let $\C\subseteq \A(r_1,r_2)$ be the abelian code with defining set $D(\C)=T(\D^*)$. Then, the set $T^{-1}\left(\Gamma\right)$ where 
 $$\Gamma=\Gamma(\C)=\left\{(i_1,i_2)\in\Z_{r_1}\times\Z_{r_2} \tq 0\leq i_1< a, 0\leq i_2<\frac{m}{a}\right\}$$
 is a set of check positions for $R^*(m-2,m)$. Furthermore, $\{0,\alpha^i\mid i\in T^{-1}\left(\Gamma\right)\}$ is an information set for $R(1,m)$ and $\{\alpha^i\mid i\notin T^{-1}\left(\Gamma\right)\}$ is an information set for $R(m-2,m)$.
\end{theorem} 
\begin{proof}
 By definition $\D^*=\Omega(1)$. We take $\Db=\U=\{1\}$. To construct the set $\Gamma(\C)$ given in (\ref{checkpositions}) we need to compute the sequences (\ref{sequence f[i]}) and (\ref{sequence gi}). By the results in Section~\ref{cyclics} we have that those sequences are obtained from (\ref{M(u)}).
 
 In this case, since $\Db=\U=\Or(1)=\{1\}$, we have a unique value
  $$M(1)=\frac{1}{\left|C_{r_1}(1)\right|}\cdot \left|C_n(1)\right|=\frac{m}{a}.$$
 Note that $C_n(1)=\{1,2,\dots,2^{m-1}\}$ and since $r_1$ divides $n$ then $a$ divides $m$.  Then, the sequences (\ref{sequence f[i]}) and (\ref{sequence gi}) are 
 $$f_1=\frac{m}{a}>f_2=0  \text{ and } g_1=m(1)=a.$$
 
 Therefore $\Gamma=\Gamma(\C)=\left\{(i_1,i_2)\in\Z_{r_1}\times\Z_{r_2} \text{ such that } f_2\leq i_2< f_1, 0\leq i_1<g_1\right\}$. 
 
 Finally, by Theorem~\ref{teorema principal check positions}, $\Gamma$ is a set of check positions for $\C$, and then $T^{-1}\left(\Gamma\right)$ is a set of check positions for $R^*(m-2,m)$. So, $\{0,\alpha^i\mid i\in T^{-1}\left(\Gamma)\right)\}$ is a set of check positions for $R(m-2,m)$ (see Remark~\ref{check positions union 0}). Since $R(m-2,m)=R(1,m)^\bot$ we are done.
\end{proof}

\begin{examples}
 The first value for $m$ that satisfies the required conditions is $m=4$. In this case, $n= 2^4-1=15, r_1=3, r_2=5, {\rm Ord}_3(2)=2$. So, let us give an information set for $R(1,4)$. By Theorem~\ref{teoremainfosetprimerorden} we have that $$\Gamma=\left\{(i_1,i_2)\in\Z_{r_1}\times\Z_{r_2} \tq 0\leq i_1< 2, 0\leq i_2< 2\right\}=\{(0,0),(1,0),(0,1),(1,1)\}.$$
 Then, by taking the isomorphism given by the Chinese Remainder Theorem, we have that $T^{-1}(\Gamma)=\{0,1,6,10\}$ and then $\{0,1,\alpha,\alpha^6,\alpha^{10}\}$ is an information set for $R(1,4)$. Moreover $\{\alpha^i\mid i\neq 0,1,6,10\}$ is an information set for $R(2,4)$. Table I shows the different information sets that we can get by using all the possible isomorphisms from $\Z_{15}$ to $\Z_{3}\times \Z_{5}$; the first column includes the image of $1\in \Z_{15}$ which determines the corresponding isomorphism, while the second column gives the set of exponents $I$ such that $\{0,\alpha^i\mid i\in I\}$ is an information set for $R(1,4)$.

\begin{table}[h]
\begin{center}
\begin{tabular}{|c|c|}
 \hline
 $\rm T(1)$&$\rm I$\\ \hline
 (1,1)&\{0,1,6,10\}\\ \hline
 (2,1)&\{0,3,5,8\}\\ \hline
 (1,2)&\{0,3,10,13\}\\ \hline
 (2,2)&\{0,6,8,10\}\\ \hline
 (1,3)&\{0,7,10,12\}\\ \hline
 (2,3)&\{0,2,5,12\}\\ \hline
 (1,4)&\{0,4,9,10\}\\ \hline
 (2,4)&\{0,5,9,14\}\\ \hline
\end{tabular}
\end{center}
\begin{center}
\caption{Information sets for R(1,4)}
\end{center}

\end{table}

 The next value for $m$ is $m=6$. In this case, $n=2^6-1=63, r_1=7, r_2=9, a=3$. So $$\Gamma=\left\{(i_1,i_2)\in\Z_{r_1}\times\Z_{r_2} \text{ such that } 0\leq i_1< 3, 0\leq i_2< 2\right\}=\{(0,0),(1,0),(2,0),(0,1),(1,1),(2,1)\}.$$
 Since $T^{-1}(\Gamma)=\{0,1,9,28,36,37\}$ one has that $\{0,1,\alpha,\alpha^9,\alpha^{28},\alpha^{36},\alpha^{37}\}$ is an information set for $R(1,6)$ and $\{\alpha^i\mid i\neq 0,1,9,28,36,37\}$ is an information set for $R(4,6)$. Again, we have taken the isomorphism given by the Chinese Remainder Theorem
 
 Finally, let us see the case $m=8$, that is, the Reed-Muller codes of length $256$. This is an interesting case because we have three possible decompositions of $n=2^8-1=255$, namely, $(r_1=3, r_2=85)$, $(r_1=5,r_2=51)$ and $(r_1=15, r_2=17)$. Table II shows the sets $T^{-1}(\Gamma)$ obtained for each decomposition; in all cases we are considering the isomorphism given by the Chinese Remainder Theorem.

\begin{table}[h]
\begin{center}
\begin{tabular}{|c|c|c|c|}
 \hline
 $r_1$&$r_2$&a& $T^{-1}(\Gamma)$\\ \hline
 3&85&2&\{0,1,3,85,87,88,171,172\}\\ \hline
 5&51&4&\{0,1,51,52,102,103,153,205\}\\ \hline
 15&17&4&\{0,1,17,18,120,136,137,153\}\\ \hline
\end{tabular}
\end{center}
\begin{center}
\caption{Information sets for R(6,8)}
\end{center}

\end{table}

\end{examples}
 
 To finish this section, we include Table III which shows the suitable values of $m$ up to length 2048. The values $m=2,3,5,7$ yield a prime number for $n=2^m-1$.

\begin{table}[h]
\begin{center}
\begin{tabular}{|c|c|c|c|c|}
 \hline
 $\mathbf{m}$&$\mathbf{n}$&$\mathbf{r_1}$&$\mathbf{r_2}$& $\mathbf{a}$\\ \hline
 4&15&3&5&2\\ \hline
 6&63&7&9&3\\ \hline
 8&255&3&85&2\\ \hline
 8&255&15&17&4\\ \hline
 8&255&5&51&4\\ \hline
 9&511&7&73&3\\ \hline
 10&1023&3&341&2\\ \hline
 10&1023&11&93&10\\ \hline
 10&1023&31&33&5\\ \hline
 11&2047&23&89&11\\ \hline
 %12&2045&7&585&3\\ \hline
\end{tabular}
\end{center}
\begin{center}
\caption{Parameters for first order RM codes up to length 2048}
\end{center}

\end{table}

\section{Information sets for second-order Reed-Muller codes}\label{second-order}

 In this section we deal with second-order Reed-Muller codes. As in the previous section, we shall apply the results of Section~\ref{cyclics}, so  we need to assume that there exist integers $r_1,r_2$ such that $n=2^m-1=r_1\cdot r_2$ and $\gcd(r_1,r_2)=1, r_1,r_2>1$. Once more, we fix a primitive $n$-th root of unity $\alpha$ and $T:\Z_n\rightarrow \Z_{r_1}\times\Z_{r_2}$ an arbitrary isomorphism of groups.

 In a similar way to the case of $R(1,m)$, we have two possibilities to get an information set for the code $R^*(2,m)$. Again, we prefer to develop the results in terms of $R^*(m-3,m)$ which yields an information set for $R(m-3,m)=R(2,m)^\bot$.
 
 Throughout this section we denote $\C^*=R^*(m-3,m)$. In this case we have that 
     $$\D^*=\Omega(1)\cup \Omega(2)$$
where $\Omega(1)=\{2^t \tq 0\leq t<m\}$ and $\Omega(2)=\{2^{t_1}+2^{t_2}\tq 0\leq t_1< t_2<m\}$.  Let $\Db\subseteq\D^*$ be a suitable set of representatives  and take $\U\subseteq\Db$ a complete set of representatives of the equivalence classes modulo $r_1$ (see (\ref{relacionU})). We will always assume that $1\in \U\subseteq \Db$ (recall that $\Omega(1)=C_n(1)$).

To get the expressions (\ref{M(u)}) we need to compute the cardinalities $\left|C_n(e)\right|$, for all $e\in\Db$, and $ \left|C_{r_1}(u)\right|$ for any element $u$ in the fixed set $\U$. As we have seen in the previous section, these computations can be made directly in the case of elements in $\Omega(1)$, however, they turn to be much more complicated when we work with the set $\Omega(2)$. So, we impose some conditions on $r_1$ in order to make the mentioned computations possible, namely, all throughout we also assume that $r_1=2^a-1$, with $a$ an integer. Then, we can sum up the restrictions on the paremeters as follows
\begin{equation}\label{restrictions}
 n=2^m-1=r_1\cdot r_2, \text{ where } r_1=2^a-1 \text{ and } \gcd(r_1,r_2)=1, r_1,r_2>1.
\end{equation} 

Note that from these conditions it follows that $a={\rm Ord}_{r_1}(2)$ and so the notation is consistent with that used in Theorem~\ref{teoremainfosetprimerorden}. In what follows we write $m=ab$.

%By Theorem~\ref{Conjunto Omega(K)} we can take as a set of restricted representatives $\overline{\D^*}=\widehat{\Omega}(1)\bigcup \widehat{\Omega}(2)$ where 
%$$\widehat{\Omega}(1)=\left\{1\right\} \text{ and } \widehat{\Omega}(2)=\left\{1+2^t\tq 0<t\leq \left\lfloor\frac{m}{2}\right\rfloor\right\}.$$      

%Now, we define $\U=\U_1\cup\U_2$ with $\U_1=\{1\}$ and $\U_2=\left\{1+2^s\tq 0<s\leq \left\lfloor\frac{a}{2}\right\rfloor\right\}$. The reader may check that such sets $\overline{\D^*}$ and $\U$ verifies the condition of Lemma~\ref{igualdadU}. 

%The following lemma shows us how to compute the cardinalities of the desired cyclotomic cosets. 
The first lemma, probably a well-known result, shows that the value of the cardinalities $\left|C_n(e)\right|$ and $\left|C_{r_1}(e)\right|$ can be easily computed for the elements in $\Omega(2)$ of the form $1+2^t$. 

\begin{lemma}\label{clases-n}
 Let $\mu,\nu$ integers such that $\mu=2^\nu-1$. Then $\left|C_\mu\left(1\right)\right|=\nu$ and for any natural number $0<t\leq \nu-1$
  $$\left|C_\mu\left(1+2^t\right)\right|=
  \begin{cases}
    \displaystyle\frac{\nu}{2} &\text{ in case }\quad \displaystyle t = \frac{\nu}{2}\\
    &\\
    \nu &\text{otherwise}
  \end{cases}$$
\end{lemma}
\begin{proof}
The equality $\left|C_\mu\left(1\right)\right|=\nu$ is clear. To see the second one, denote $\delta=\left|C_\mu\left(1+2^t\right)\right|$. Then $2^\delta\cdot (1+2^t)=2^\delta+2^{\delta+t}\equiv 1+2^t$ modulo $\mu$. So, either $\delta=\nu$ (and $\delta+t\equiv t$ mod $\nu$) or $\delta+t=\nu$ and $\delta=t$. The second condition implies $\delta=t=\nu/2$. 
\end{proof}

As we have observed in the previous section the elements in the suitable set $\Db$ might not satisfy condition (\ref{acotados}). However, we can relate any element in $\Db$ with another one in the same $2$-cyclotomic coset modulo $n$ that satisfies that condition (see Theorem~\ref{teoremaacotados}). The following result details this relationship.

\begin{proposition}\label{epsilon}
 Let $\Db$ be a suitable set of representatives. There exists a bijection  $\varepsilon:\Db\setminus\{1\}\longrightarrow \{s\in \Z\mid 1\leq s\leq \left\lfloor \frac{m}{2}\right\rfloor\}$ where, for any $e\in\Db\setminus\{1\}$,  $\varepsilon(e)$ is the unique integer such that $1+2^{\varepsilon(e)}$ belongs to $C_n(e)$ and satisfies condition (\ref{acotados}).
\end{proposition}
\begin{proof}
 By Theorem~\ref{teoremaacotados}, for any $e\in\Db\setminus\{1\}$ there exists $e'=1+2^s$, where $0<s\leq \left\lfloor \frac{m}{2}\right\rfloor$, such that $C_n(e')= C_n(e)$. We define $\varepsilon(e)=s$ and we are going to prove that $e'$ is unique in $C_n(e)$ satisfying condition (\ref{acotados}). 
  
  Suppose that exists $e''=1+2^{s'}\in C_n(e),\ e''\neq e'$, satisfying the required condition. Then there must exist $v\in\{1,\dots,m-1\}$ such that $e''\equiv 2^v e'\equiv 2^v+2^{v+s}$ mod $n$. Necessarily this implies $v+s=m, v=s'$. So, $s'=m-s\geq m-\left\lfloor m/2 \right\rfloor\geq \left\lfloor m/2\right\rfloor$; since $e''\neq e'$ then $s'> \left\lfloor m/2\right\rfloor$, a contradiction. So, we can say that $\varepsilon$ is well-defined. Now, for any $s\in\{1,\dots,\left\lfloor \frac{m}{2}\right\rfloor\}$ the element $1+2^s$ belongs to $\Omega(2)$ and so there exists $e\in\Db\setminus\{1\}$ such that $C_n(e)=C_n(1+2^s)$. Moreover, if $e\neq e'\in\Db$ then $C_n(e)\neq C_n(e')$ and so $\varepsilon(e)\neq \varepsilon(e')$. This implies that $\varepsilon$ is a bijection.
\end{proof}

\begin{remark}\label{formulaepsilon}
 Observe that for all $e=2^{t_1}+2^{t_2}\in\Omega(2)$ one has that $1+2^\delta$, with $\delta=\min\{t_2-t_1,m-t_2+t_1\}$, belongs to $C_n(e)$ and verifies (\ref{acotados}), therefore $\varepsilon(e)=\min\{t_2-t_1,m-t_2+t_1\}$.
\end{remark}

Under the notation introduced by Proposition~\ref{epsilon} we can say that for any $e\in\Omega(2)$ the element $1+2^{\varepsilon(e)}$ is the unique element in $C_n(e)$ satisfying condition (\ref{acotados}).

Now, from Proposition~\ref{epsilon} and Lemma~\ref{clases-n} we obtain the next result which gives us the desired cardinalities.

\begin{proposition}\label{cardinalities}
Let $n=2^m-1=r_1\cdot r_2$ satisfying (\ref{restrictions}), that is,  $m=ab$, with  $a$ such that  $r_1=2^a-1$  and $(r_1,n/r_1)=1$. For any $e\in\Omega(2)$ one has that
\begin{enumerate}%[a)]
 \item  $\left|C_n(e)\right|=
  \begin{cases}
    \displaystyle\frac{m}{2} &\text{ in case }\quad \displaystyle \varepsilon(e)= \frac{m}{2}\\
    &\\
    m &\text{otherwise}
  \end{cases}$
 \vspace{0.4cm}
  
 \item  $\left|C_{r_1}(e)\right|=
  \begin{cases}
    \displaystyle\frac{a}{2} &\text{ in case }\quad \displaystyle \varepsilon(e)\equiv \frac{a}{2} \text{\rm{ mod }} a\\
    &\\
    a &\text{otherwise}
  \end{cases}$
\end{enumerate} 
\end{proposition}
\begin{proof}
 Let $e=2^{t_1}+2^{t_2}\in\Omega(2)$ and $e'=1+2^{\varepsilon(e)}$. Then $C_n(e)=C_n(e')$. Moreover $\varepsilon(e)=t_2-t_1$ or $\varepsilon(e)=m-(t_2-t_1)$ (see Remark~\ref{formulaepsilon}). By Lemma~\ref{clases-n} $\left|C_n(e')\right|=m/2$ if $\varepsilon(e)=m/2$ and $\left|C_n(e`)\right|=m$ otherwise, which yields \textit{1)}.
 
 To prove \textit{2)} note that there exists an integer $v\in\{0,\dots,m-1\}$ such that $e'\equiv 2^v e$ modulo $n$. Since $n=r_1\cdot r_2$ we have that $e'\equiv 2^v e$ modulo $r_1$. Then $C_{r_1}(e)=C_{r_1}(e')$ and so $\left|C_{r_1}(e)\right|=\left|C_{r_1}(e')\right|$. Note also  that $e'\equiv 1+2^\mu$ mod $r_1$ where $\mu$ is the residue of $\varepsilon(e)$ modulo $a$. By Lemma~\ref{clases-n} $\left|C_{r_1}(e')\right|=a/2$ if $\varepsilon(e)\equiv a/2$ mod $a$ and $\left|C_{r_1}(e')\right|=a$ otherwise. Finally, observe that if $\varepsilon(e)\equiv 0$ mod $a$ then $e'\in C_{r_1}(1)$ and so $\left|C_{r_1}(e')\right|=\left|C_{r_1}(1)\right|=a$. This finishes the proof.
\end{proof}

Now, all that is required to obtain the expressions (\ref{M(u)}) is the description of the sets $\Or(u)$ with $u\in \U$, that is, $\Or(1)$ and $\Or(e)$ with $e\in \Omega(2)\cap\U$. The next results solve this problem. The first one gives the information we need about $\Or(1)$. Recall that all throughout $\Db$ denotes a suitable set of representatives.

\begin{proposition}\label{proposicionOde1}
 \begin{enumerate}
  \item  $\Or(1)=\{1\}\cup \{e\in\Db\mid \varepsilon(e)=\lambda a \text{ with } 1\leq \lambda \leq \left\lfloor \frac{b}{2}\right\rfloor\}$.
  \item $\left|\Or(1)\right|=1+\left\lfloor b/2\right\rfloor$.
  %\item $\left|C_{r_1}(1)\right|=a$.
  \item If $b$ is even then $\varepsilon^{-1}(m/2)\in\Or(1)$, $\left|C_n(\varepsilon^{-1}(m/2))\right|=m/2$, and $\left|C_n(e)\right|=m$ for any $e\in\Or(1)\setminus\{\varepsilon^{-1}(m/2)\}$.
  \item If $b$ is odd then $\left|C_n(e)\right|=m$ for any $e\in\Or(1)$.
 \end{enumerate}
\end{proposition}

\begin{proof}
 Let $e\in\Db, e\neq 1$, such that $e\equiv 1$ mod $r_1$. Then $C_{r_1}(1+2^{\varepsilon(e)})=C_{r_1}(e)=C_{r_1}(1)$. Let $\mu$ be the residue of $\varepsilon(e)$ modulo $a$, then $C_{r_1}(1+2^{\varepsilon(e)})=C_{r_1}(1+2^\mu)=C_{r_1}(1)$. So $1+2^\mu$ and $1$ are different elements in $\{1,\dots,r_1-1\}$ that belongs to the same $2$-cyclotomic coset modulo $r_1$; this is only possible in the case $\mu=0$. Therefore $\varepsilon(e)\equiv 0$ mod $a$. Let $\varepsilon(e)=\lambda a$ with $\lambda$ a natural number. It is easy to see that this implies $1\leq \lambda\leq\left\lfloor \frac{b}{2}\right\rfloor$ because $\varepsilon(e)\leq \left\lfloor \frac{m}{2}\right\rfloor$ (recall that $n=2^m-1$ where $m=ab$). Conversely, if $e\in\Db$ and $\varepsilon(e)=\lambda a$ then $C_{r_1}(e)=C_{r_1}(1+2^{\lambda a})=C_{r_1}(2)=C_{r_1}(1)$. Since $\Db$ is a suitable set of representatives we conclude that $e\equiv 1$ modulo $r_1$. This proves \textit{1)}.
 
 Statement \textit{2)} is immediate from \textit{1)} and Proposition~\ref{epsilon}. % \textit{3)} follows from the equality $r_1=2^a-1$.
 
 Now, suppose that $b$ is even. Then $\frac{m}{2}=a\cdot\frac{b}{2}$, so $\varepsilon^{-1}(m/2)\in\Or(1)$ by \textit{1)}. The last part of \textit{3)} follows from Proposition~\ref{cardinalities}.
 
 Finally, suppose that $b$ is odd. Then, in case $m$ even,  $\frac{m}{2}\equiv\frac{a}{2}$ modulo $a$, so $\varepsilon^{-1}(m/2)\notin\Or(1)$. Therefore, $\left|C_n(e)\right|=m$ for any $e\in\Or(1)$ by Proposition~\ref{cardinalities}. This finishes the proof. 
\end{proof}

\begin{remark} \label{elegirU}
 As we have already seen we assume that $1\in\U\subseteq\Db$. Furthermore, depending on the parity of $a$ and $b$ we will take the following elements in $\U$ by convention:
 \begin{enumerate}
  \item If $m$ is even we also assume that $\varepsilon^{-1}(m/2)\in \U$  unless $\varepsilon^{-1}(m/2)\in\Or(1)$.
  \item If $a$ is even we assume that $\varepsilon^{-1}(a/2)\in \U$   unless $\varepsilon^{-1}(a/2)\in\Or(\varepsilon^{-1}(m/2))$. Observe that $\varepsilon^{-1}(a/2)$ never belongs to $\Or(1)$.
 \end{enumerate}

\end{remark}

\vspace{.2cm}

The next result yields the description of the set $\Or(e)$ for any $e\in\U\setminus\{1\}$. Let us observe that $e\in \U\setminus\{1\}$ implies $e\in\U\cap \Omega(2)$.
\begin{proposition}\label{proposicionOdee}
 Let $e\in \U\setminus\{1\}$. Then $\Or(e)=\{e'\in\Db\mid \varepsilon(e')\equiv \varepsilon(e) \text{ mod } a \text{ \rm{or} } \varepsilon(e')\equiv a-\varepsilon(e) \text{ mod } a\}$.
\end{proposition}
\begin{proof}
 Let $e'\in \Db$ such that $e'\equiv e$ mod $r_1$. Then $C_{r_1}(1+2^{\varepsilon(e)})=C_{r_1}(1+2^{\varepsilon(e')})$. Let $\mu$ and $\mu'$ be the residues of $\varepsilon(e)$ and $\varepsilon(e')$ modulo $a$ respectively. So $C_{r_1}(1+2^\mu)=C_{r_1}(1+2^{\mu'})$, that is, $1+2^\mu$ and $1+2^{\mu'}$ are elements in $\{1,\dots,r_1-1\}$ that belong to the same $2$-cyclotomic coset modulo $r_1$. This is only possible in the cases $\mu=\mu'$ or $\mu=a-\mu'$. 
 
 Conversely, suppose that $e'\in\Db$ and $\varepsilon(e')\equiv\varepsilon(e)$ modulo $a$. Then, $C_{r_1}(e)=C_{r_1}(1+2^{\varepsilon(e)})=C_{r_1}(1+2^{\varepsilon(e')})=C_{r_1}(e')$. Since $\Db$ is a suitable set of representatives this implies $e'\equiv e$ modulo $r_1$. Finally if $e'\in\Db$ and $\varepsilon(e')\equiv a-\varepsilon(e)$ modulo $a$ we have the following sequence of equalities: $C_{r_1}(e')=C_{r_1}(1+2^{\varepsilon(e')})=C_{r_1}(1+2^{a-\varepsilon(e)})=C_{r_1}(1+2^{\varepsilon(e)})=C_{r_1}(e)$. Again, we conclude that $e'\equiv e$ modulo $r_1$. So we are done. 
\end{proof}

The following propositions complete the information about the sets $\Or(e)$, with $e\in\U\setminus\{1\}$, by making a distinction between the cases $b$ even and $b$ odd ($m=ab$).

\begin{proposition}\label{proposicionOdeeEVEN}
 Let $e\in\U\setminus\{1\}$ and suppose that $b$ is even. Then
 \begin{enumerate}
  \item For any $e'\in\Or(e)$ one has that $\left|C_n(e')\right|=m$.
  \item If $a$ is odd then $\left|\Or(e)\right|=b$ and $\left|C_{r_1}(e)\right|=a$.
  \item If $a$ is even then in case $e=\varepsilon^{-1}(a/2)$ one has that $\left|\Or(e)\right|=b/2$, $\left|C_{r_1}(e)\right|=a/2$ and otherwise $\left|\Or(e)\right|=b$, $\left|C_{r_1}(e)\right|=a$.  
 \end{enumerate}
\end{proposition}
\begin{IEEEproof} 
 First of all note that since $b$ is even then $\frac{m}{2}\equiv a$ modulo $a$ which implies $\varepsilon^{-1}(m/2)\in\Or(1)$ (see Proposition~\ref{proposicionOde1}), so $e'\neq \varepsilon^{-1}(m/2)$ for all $e'\in\Or(e)$. This proves \textit{1)} by Proposition~\ref{cardinalities}.
 
 In the rest of proof we denote by $\mu$ the residue of $\varepsilon(e)$ modulo $a$. Observe that $\mu\neq 0$ because $e\in\U\setminus\{1\}$. Let us define $A=\{e'\in\Db\mid \varepsilon(e')\equiv \varepsilon(e) \text{ mod } a\}$ and $B=\{e'\in\Db\mid \varepsilon(e')\equiv a-\varepsilon(e) \text{ mod } a\}$. By Proposition~\ref{proposicionOdee}, $\Or(e)=A\cup B$.  Note also that $A\cap B\neq \emptyset$ if and only if $\mu=a/2$, that is, if and only if $\Or(e)=A=B=A\cap B$.
 
 Now we write $A=\{e'\in\Db\mid \varepsilon(e')=\lambda a+\mu, \text{  with } \lambda \text{ a natural number}\}$. To get the cardinality of $A$ we need to study the range of values of $\lambda$. By definition of $\varepsilon$ one has that $1\leq \varepsilon(e')\leq \left\lfloor \frac{m}{2}\right\rfloor$, and since
 $$\displaystyle\frac{b}{2} a+\mu=\frac{m}{2}+\mu>\frac{m}{2}\geq\left\lfloor \frac{m}{2}\right\rfloor$$
and
 $$\displaystyle \left(\frac{b}{2}-1\right)a+\mu=\frac{m}{2}-a+\mu<\frac{m}{2},$$
 we conclude that $\left|A\right|=\left|\{\lambda\in \Z\mid 0\leq \lambda< \frac{b}{2}\}\right|=b/2$. An analogous argument yields $\left|B\right|=b/2$.  
 
 To prove \textit{2)} we assume that $a$ is odd. Then,  $\left|\Or(e)\right|=\left|A\right|+\left|B\right|$, that is, $\left|\Or(e)\right|=\frac{b}{2}+\frac{b}{2}=b$. The equality $\left|C_{r_1}(e)\right|=a$ follows from Proposition~\ref{cardinalities}.
 
 Finally, we suppose that $a$ is even. Then $\varepsilon^{-1}(a/2)\in\U$ (see Remark~\ref{elegirU}). If $e=\varepsilon^{-1}(a/2)$ then $\left|\Or(e)\right|=\left|A\right|=b/2$ and $\left|C_{r_1}(e)\right|=a/2$. If $e\neq\varepsilon^{-1}(a/2)$ then, by repeating the arguments of the previous paragraph, we obtain $\left|\Or(e)\right|=b$ and $\left|C_{r_1}(e)\right|=a$.
\end{IEEEproof} 

\begin{proposition}\label{proposicionOdeeODD}
Let $e\in\U\setminus\{1\}$ and suppose that $b$ is odd. Then $\left|C_n(e')\right|=m$ for all $e'\in\Or(e)\setminus\{e\}$ and
 \begin{enumerate}
  \item If $a$ is odd then $\left|\Or(e)\right|=b$, $\left|C_{r_1}(e)\right|=a$ and $\left|C_n(e)\right|=m$.
  \item If $a$ is even then
   \begin{enumerate}
    \item If $e=\varepsilon^{-1}(m/2)$ then $\left|\Or(e)\right|=(b+1)/2$, $\left|C_n(e)\right|=m/2$ and $\left|C_{r_1}(e)\right|=a/2$.
    \item If $e\neq \varepsilon^{-1}(m/2)$ then $\left|\Or(e)\right|=b$, $\left|C_n(e')\right|=m$ and $\left|C_{r_1}(e)\right|=a$.
   \end{enumerate}
 \end{enumerate}
\end{proposition}

\begin{IEEEproof}
  
  We denote by $\mu$ the residue of $\varepsilon(e)$ modulo $a$ ($\mu> 0$). We are going to use a similar technique to that was used in the proof of the previous proposition. We define $A=\{e'\in\Db\mid \varepsilon(e')\equiv \varepsilon(e) \text{ mod } a\}$ and $B=\{e'\in\Db\mid \varepsilon(e')\equiv a-\varepsilon(e) \text{ mod } a\}$. Then $\Or(e)=A\cup B$. Observe that $A\cap B\neq \emptyset$ if and only if $\mu=a/2$ if and only if $\Or(e)=A=B=A\cap B$. Then
   $$\displaystyle\left\lfloor \frac{b}{2}\right\rfloor a+\mu=\frac{b-1}{2} a+\mu=\frac{m}{2}-\frac{a}{2}+\mu\ ;$$
   this value is less than or equal to $\left\lfloor\frac{m}{2} \right\rfloor$ if and only if $\mu\leq \left\lfloor \frac{a}{2}\right\rfloor$. So $\left|A\right|=\left|\{0\leq\lambda\leq\left\lfloor \frac{b}{2}\right\rfloor\}\right|=\left\lfloor \frac{b}{2}\right\rfloor+1$ in case $\mu\leq \left\lfloor \frac{a}{2}\right\rfloor$ and $\left|A\right|=\left|\{0\leq\lambda<\left\lfloor \frac{b}{2}\right\rfloor\}\right|=\left\lfloor \frac{b}{2}\right\rfloor$ otherwise.
   
   To compute the cardinality of $B$ we observe that 
   $$\displaystyle\left\lfloor \frac{b}{2}\right\rfloor a+a-\mu=\frac{b-1}{2} a+a-\mu=\frac{m}{2}+\frac{a}{2}-\mu.$$
   It is easy to see that this value is less than or equal to $\left\lfloor \frac{m}{2}\right\rfloor$ if and only if $\mu\geq \left\lceil \frac{a}{2}\right\rceil$. Then $\left|B\right|=\left\lfloor \frac{b}{2}\right\rfloor+1$ in case $\mu\geq \left\lceil \frac{a}{2}\right\rceil$ and $\left|B\right|=\left\lfloor \frac{b}{2}\right\rfloor$ otherwise.
   
Now, suppose that $a$ is odd. Then 
$$\left\lceil \frac{a}{2}\right\rceil=\frac{a}{2}+\frac{1}{2},\ \left\lfloor \frac{a}{2}\right\rfloor=\frac{a}{2}-\frac{1}{2}$$ 
and so $\mu\leq\left\lfloor \frac{a}{2}\right\rfloor$ if and only if $\mu<\left\lceil \frac{a}{2}\right\rceil$. Therefore, $\left|\Or(e)\right|=\left|A\right|+\left|B\right|=2\left\lfloor \frac{b}{2}\right\rfloor+1=b$. On the other hand, since $a$ and $m$ are odd, for all $e'\in\Or(e)$ one has that $\left|C_n(e')\right|=m, \left|C_{r_1}(e')\right|=a$, by Proposition~\ref{cardinalities}. This proves \textit{1)}.

Finally, suppose that $a$ is even. Note that since $b$ is odd we have $\frac{m}{2}\equiv\frac{a}{2}$ modulo $a$, so $\varepsilon^{-1}(m/2)\in\U$ and $\varepsilon^{-1}(a/2)\in\Or(\varepsilon^{-1}(m/2))$ (see Remark~\ref{elegirU} and Proposition~\ref{proposicionOdee}).

First, assume that $e=\varepsilon^{-1}(m/2)$. Then $\mu=a/2$ and $\left|\Or(e)\right|=\left|A\right|=\left|B\right|=\left\lfloor \frac{b}{2}\right\rfloor+1=(b+1)/2$. Moreover, we have $\left|C_{r_1}(e)\right|=a/2, \left|C_n(e)\right|=m/2$ and $\left|C_n(e')\right|=m$ for any $e'\in\Or(e)\setminus\{e\}$. This proves \textit{2 a)}.

In the case $e\neq\varepsilon^{-1}(m/2)$ we have that $\mu\neq\frac{a}{2}$ because $e,\varepsilon^{-1}(m/2)\in\U$ and they satisfy $C_{r_1}(e)=C_{r_1}(1+2^\mu)$ and $C_{r_1}(\varepsilon^{-1}(m/2))=C_{r_1}(1+2^{a/2})$. Note that $\mu\neq\frac{a}{2}$ implies again $\mu\leq\left\lfloor \frac{a}{2}\right\rfloor=\frac{a}{2}$ if and only if $\mu<\left\lceil \frac{a}{2}\right\rceil=\frac{a}{2}$. So $\left|\Or(e)\right|=\left|A\right|+\left|B\right|=2\left\lfloor \frac{b}{2}\right\rfloor+1=b$. In addition, $\left|C_{r_1}(e)\right|=a$ and $\left|C_{n}(e')\right|=m$ for all $e'\in\Or(e)$. 
\end{IEEEproof} 

The last result we need before enunciating our main theorem gives us the cardinality of any choice of the set $\U$. It uses the mentioned concept of 2-weight of an integer (see paragraph before Definition~\ref{defRM}). We talk about the 2-weight of a 2-cyclotomic coset when one of its elements (and so all of them) has that 2-weight.

\begin{lemma}\label{cardinaldeU}
For any election of the set of representatives $\U$ one has that $$\left|\U\right|=1+\left|\{C_{r_1}(e)\tq e\in\U\cap \Omega(2)\}\right|=1+\left\lfloor a/2\right\rfloor.$$
\end{lemma}
\begin{proof}
 First of all, as we have noted in Remark~\ref{elegirU} we always assume that $1\in\U$, so we restrict our attention to the cardinality of $\U\setminus\{1\}$. 
 We define a map 
 $$\varphi:\U\setminus\{1\}\rightarrow\{C_{r_1}(e)\tq e\in\U\cap \Omega(2)\}$$
 given by $\varphi(e)=C_{r_1}(e)$. We are going to see that $\varphi$ is a bijection. 
 
 Let $e\in\U\setminus\{1\}$. Then $C_{r_1}(e)=C_{r_1}(1+2^{\varepsilon(e)})$. Note that, by Proposition~\ref{proposicionOde1}, the case $\varepsilon(e)\equiv 0$ modulo $a$ implies $e\in\Or(1)$ which yields a contradiction, so $C_{r_1}(e)$ is a coset of weight $2$ and $\varphi$ is well-defined.
 
 Now, let $e',e\in\U\setminus\{1\}, e\neq e'$. By the definitions of suitable set of representatives and $\U$ we have that $C_{r_1}(e)\neq C_{r_1}(e')$, so $\varphi$ is injective. Take $C_{r_1}(e)$ such that $e\in \U\cap \Omega(2)$. Then $e\neq 1$ because $e\in \Omega(2)$, so $C_{r_1}(e)=\varphi(e)$ which implies that $\varphi$ is surjective.

 Finally, to see the right hand equality take any integer $s\in\{1,\dots,\left\lfloor \frac{a}{2}\right\rfloor\}$. Then $1+2^s\in\Omega(2)\subset \D^*$. Since $T_1(1+2^s)=1+2^s$ and $T(\Db)$ is a restricted set of representatives, there exists $e\in \Db\cap\Omega(2)$ such that $C_{r_1}(e)=C_{r_1}(1+2^s)$; moreover, in case $e\notin \U$, there must exist $u\in \U\setminus\{1\}$ such that $u\equiv e$ mod $r_1$ and then $C_{r_1}(u)=C_{r_1}(e)=C_{r_1}(1+2^s)$. Finally, if $u\in\U\cap \Omega(2)$ then there exists $s_u\in\{1,\dots,\left\lfloor \frac{a}{2}\right\rfloor\}$ such that $1+2^{s_u}\in C_{r_1}(u)$, and if $u\neq u'$ then $s_u\neq s_{u'}$ because $T(\Db)$ is a restricted set of representatives. This finishes the proof. 
\end{proof}

\vspace{0.2cm}

Now, we can present the  main result for second-order Reed-Muller codes. 

\begin{theorem}\label{teoremainfosetsegundoorden}
 Let $\C^*=R^*(m-3,m)$. Suppose that $n=2^m-1=r_1\cdot r_2$, where $m=ab$, $r_1=2^a-1$ and $\gcd(r_1,r_2)=1, r_1,r_2>1$. Let $\D^*\subseteq \Z_n$ be the defining set of $\C^*$, with respect to $\alpha$, a primitive $n$-th root of unity, and $\C\subseteq \A(r_1,r_2)$ the abelian code with defining set $D(\C)=T(\D^*)$. Then, the values of $f_i$ and $g_i$ from (\ref{sequence f[i]}) and (\ref{sequence gi}) are
  $$f_1=b^2, f_2=\frac{b(b+1)}{2}\quad \text{ and }\quad g_1=\frac{a(a-1)}{2}, g_2=\frac{a(a+1)}{2},$$
  respectively.
  
  Therefore the set $T^{-1}\left(\Gamma\right)$ where $\Gamma=\Gamma(\C)=\left\{(i_1,i_2)\in\Z_{r_1}\times\Z_{r_2} \text{ such that }\right.$

  $$\left.\begin{array}{l}
  \displaystyle 0\leq i_1<\frac{a(a-1)}{2}\ \text{ and }\ 0\leq i_2<b^2 \\
  \\
  \text{ or }\\
  \\
  \displaystyle\frac{a(a-1)}{2}\leq i_1<\frac{a(a+1)}{2}\ \text{ and }\ 0\leq i_2<\frac{b(b+1)}{2} 
  \end{array}\right\}$$
 is a set of check positions for $R^*(m-3,m)$. Furthermore, $\{0,\alpha^i\mid i\in T^{-1}\left(\Gamma\right)\}$ is an information set for $R(2,m)$ and $\{\alpha^i\mid i\notin T^{-1}\left(\Gamma\right)\}$ is an information set for $R(m-3,m)$.
\end{theorem}
\begin{IEEEproof}
 Let $\C\subseteq \A(r_1,r_2)$ be the abelian code with defining set $D(\C)=T(\D^*)$. In order to get the set $\Gamma(\C)$ we are going to compute the sequences (\ref{sequence f[i]}) and (\ref{sequence gi}) by using the expressions (\ref{M(u)}). To do that we have to study different cases depending on the parity of $a$ and $b$ respectively. As we will see, in any case we will obtain the same set of check positions.
 \begin{itemize}
  \item Case $m$ odd ($a, b$ odd). %In this case $\varepsilon^{-1}(m/2), \varepsilon^{-1}(a/2) \notin\U$ obviously.
  
  Take $1\in\U$. By applying Proposition~\ref{proposicionOde1} we have that
  $$\displaystyle M(1)=\frac{1}{\left|C_{r_1}(1)\right|}\sum\limits_{v\in\Or(1)}\left|C_n(v)\right|=\frac{1}{a}\sum\limits_{v\in\Or(1)}\left|C_n(v)\right|=\frac{1}{a}\cdot m\cdot \left(\left\lfloor \frac{b}{2}\right\rfloor+1\right)=b\cdot \frac{b+1}{2}=\frac{b(b+1)}{2}.$$
  For any $e\in\U\setminus\{1\}$ we use Proposition~\ref{proposicionOdeeODD} \textit{1)} and we obtain
  $$M(e)=\frac{1}{\left|C_{r_1}(e)\right|}\sum\limits_{v\in\Or(e)}\left|C_n(v)\right|=\frac{1}{a}\cdot m\cdot b=b^2.$$
  Therefore $f_1=b^2>f_2=\frac{b(b+1)}{2}$. By definition, the sequence $g_1<g_2$ is obtained as follows
  $$g_1=\sum\limits_{M(v)=f_1}\left|C_{r_1}(v)\right|=a\cdot \left|\U\setminus\{1\}\right|=a\cdot \left\lfloor \frac{a}{2}\right\rfloor=a\cdot\frac{a-1}{2},$$
   and
  $$g_2=\sum\limits_{M(v)\geq f_2}\left|C_{r_1}(v)\right|=a\cdot \left|\U\right|=a\cdot\frac{a+1}{2},$$
  where we have used Lemma~\ref{cardinaldeU}.
  
  \item Case $a$ odd and $b$ even. In this case $\varepsilon^{-1}(m/2)\in\Or(1)$ (see Proposition~\ref{proposicionOde1} \textit{4)}). %Again, we have that $\varepsilon^{-1}(m/2), \varepsilon^{-1}(a/2) \notin\U$.
  
  Take $1\in\U$. By Proposition~\ref{proposicionOde1}
    $$\begin{array}{rcl}     M(1)&=&\displaystyle\frac{1}{\left|C_{r_1}(1)\right|}\sum\limits_{v\in\Or(1)}\left|C_n(v)\right|=\frac{1}{a}\left(\left|C_n(\varepsilon^{-1}(m/2))\right|+\sum\limits_{v\in\Or(1)\setminus\{ \varepsilon^{-1}(m/2)\}}\left|C_n(v)\right|\right)=\\
    &=&\displaystyle\frac{1}{a}\left(\frac{m}{2}+\frac{b}{2}\cdot m\right)=\frac{b}{2}+\frac{b^2}{2}=\frac{b(b+1)}{2}.
    \end{array}$$
    
  On the other hand, for any $e\in\U\setminus\{1\}$ one has that 
  $$M(e)=\frac{1}{\left|C_{r_1}(e)\right|}\sum\limits_{v\in\Or(e)}\left|C_n(v)\right|=\frac{1}{a}\cdot b\cdot m=b^2.$$
  Therefore, $f_1=b^2>f_2=\frac{b(b+1)}{2}$. The sequence $g_1<g_2$ is 
  $$g_1=\sum\limits_{M(v)=f_1}\left|C_{r_1}(v)\right|=\sum\limits_{u\in\U\setminus\{1\}}\left|C_{r_1}(u)\right|= \frac{a-1}{2}\cdot a=\frac{a(a-1)}{2},$$
$$g_2=g_1+\left|C_{r_1}(1)\right|=\frac{a(a-1)}{2}+a=\frac{a(a+1)}{2},$$

\item Case $a$ and $b$ even. %As in the previous case $\varepsilon^{-1}(m/2)\in\Or(1)$, however now $\varepsilon^{-1}(a/2)\in\U$ (see Remark~\ref{elegirU}).

Take $1\in\U$. Then $M(1)$ is computed exactly as in the previous case, so $M(1)=\frac{b(b+1)}{2}$.

Now, consider $\varepsilon^{-1}(a/2)\in\U$. Then, by Proposition~\ref{proposicionOdeeEVEN} \textit{3)}
$$M(\varepsilon^{-1}(a/2))=\frac{1}{\left|C_{r_1}(\varepsilon^{-1}(a/2))\right|}\sum\limits_{v\in\Or(\varepsilon^{-1}(a/2))}\left|C_n(v)\right|=\frac{2}{a}\cdot \frac{b}{2}\cdot m=b^2.$$

Finally, for any $e\in\U\setminus\{1,\varepsilon^{-1}(a/2)\}$ one has
  $$M(e)=\frac{1}{\left|C_{r_1}(e)\right|}\sum\limits_{v\in\Or(e)}\left|C_n(v)\right|=\frac{1}{a}\cdot b\cdot m=b^2.$$
  
  We conclude $f_1=b^2>f_2=\frac{b(b+1)}{2}$. The values $g_1, g_2$ are
    $$g_1=\sum\limits_{M(v)=f_1}\left|C_{r_1}(v)\right|=a\cdot \left(\left|\U\right|-2\right)+\left|C_{r_1}(\varepsilon^{-1}(a/2))\right|=a\cdot \left(\frac{a}{2}-1\right)+\frac{a}{2}=\frac{a(a-1)}{2},$$
   and
 $$g_2=g_1+\left|C_{r_1}(1)\right|=\frac{a(a-1)}{2}+a=\frac{a(a+1)}{2},$$

\item Case $a$ even and $b$ odd. In this case $\varepsilon^{-1}(m/2)\notin\Or(1)$ and $\varepsilon^{-1}(m/2)\in\U$ (see Remark~\ref{elegirU}). Observe that $\varepsilon^{-1}(a/2)\in\Or(\varepsilon^{-1}(m/2))$ and so $\varepsilon^{-1}(a/2)\notin\U$.

For $1\in\U$ we have
  $$\displaystyle M(1)=\frac{1}{\left|C_{r_1}(1)\right|}\sum\limits_{v\in\Or(1)}\left|C_n(v)\right|=\frac{1}{a}\cdot\left|\Or(1)\right|\cdot m=\frac{1}{a}\cdot\frac{b+1}{2}\cdot m=\frac{b(b+1)}{2}.$$
  
Now consider $\varepsilon^{-1}(m/2)\in\U$. Then, by Proposition~\ref{proposicionOdeeODD} \textit{2) a)}  
$$M(\varepsilon^{-1}(m/2))=\frac{1}{\left|C_{r_1}(\varepsilon^{-1}(m/2))\right|}\sum\limits_{v\in\Or(\varepsilon^{-1}(m/2))}\left|C_n(v)\right|=\frac{2}{a}\cdot\left[\left(\frac{b+1}{2}-1\right)\cdot m+\frac{m}{2}\right]=(b-1)\cdot b+b=b^2.$$

Finally, for any $e\in\U\setminus\{1,\varepsilon^{-1}(m/2)\}$ we use Proposition~\ref{proposicionOdeeODD} \textit{2) b)} and we obtain $M(e)=\frac{1}{a}\cdot b\cdot m=b^2$. Therefore, as in all the previous cases $f_1=b^2>f_2=\frac{b(b+1)}{2}$, and once more 
 $$g_1=\left|C_{r_1}(\varepsilon^{-1}(m/2)\right|+\left(\frac{a}{2}-1\right)\cdot a=\frac{a}{2}+\frac{a-2}{2}\cdot a=\frac{a(a-1)}{2}$$
and
 $$g_2=g_1 + \left|C_{r_1}(1)\right|=\frac{a(a-1)}{2}+a=\frac{a(a+1)}{2}.$$
 
\end{itemize}

In conclusion, all the cases yield the same sequences $f_1>f_2$, $g_1<g_2$. Finally, by Theorem~\ref{teorema principal check positions}, $\Gamma(\C)$ is a set of check positions for $\C$, and then $T^{-1}\left(\Gamma\right)$ is a set of check positions for $R^*(m-3,m)$. So, $\{0,\alpha^i\mid i\in T^{-1}\left(\Gamma)\right)\}$ is a set of check positions for $R(m-3,m)$ (see Remark~\ref{check positions union 0}). Since $R(m-3,m)=R(2,m)^\bot$ we are done.

\end{IEEEproof}

\begin{examples}
 The first value for $m$ that satisfies conditions (\ref{restrictions}) is $m=4$. In this case, $n= 2^4-1=15, r_1=3, r_2=5, a=b=2 $. So, let us give an information set for $R(2,4)$. By Theorem~\ref{teoremainfosetsegundoorden} we have that $$\Gamma=\{(0,0),(1,0),(2,0),(0,1),(1,1),(2,1),(0,2),(1,2),(2,2),(0,3)\}.$$
 Then, by taking the isomorphism given by the Chinese Remainder Theorem, we have that $T^{-1}(\Gamma)=\{0,1,2,3,5,6,7,10,11,12\}$. So $\{0,\alpha^i\mid i\in T^{-1}\left(\Gamma\right)\}$ is an information set for $R(2,4)$ and $\{\alpha^i\mid i\notin T^{-1}\left(\Gamma\right)\}$ is an information set for $R(1,4)$. In the same sense than in the previous section, Table IV shows the different information sets that we can get by using all the possible isomorphisms from $\Z_{15}$ to $\Z_{3}\times \Z_{5}$; the first column includes the image of $1\in \Z_{15}$ which determines the corresponding isomorphism, while the second column gives the set of exponents $I$ such that $\{0,\alpha^i\mid i\in I\}$ is an information set for $R(2,4)$. From these information sets we can obtain the corresponding ones for $R(1,4)$; the reader may check that we get four new information sets with respect to that obtained in Table I.

\begin{table}[h]
\begin{center}
\begin{tabular}{|c|c|}
 \hline
 $\rm T(1)$&$\rm I$\\ \hline
 (1,1)&\{0,1,2,3,5,6,7,10,11,12\}\\ \hline
 (2,1)&\{0,1,2,3,5,6,7,10,11,12\}\\ \hline
 (1,2)&\{0,1,3,5,6,8,9,10,11,13\}\\ \hline
 (2,2)&\{0,1,3,5,6,8,9,10,11,13\}\\ \hline
 (1,3)&\{0,2,4,5,6,7,9,10,12,14\}\\ \hline
 (2,3)&\{0,2,4,5,6,7,9,10,12,14\}\\ \hline
 (1,4)&\{0,3,4,5,8,9,10,12,13,14\}\\ \hline
 (2,4)&\{0,3,4,5,8,9,10,12,13,14\}\\ \hline
\end{tabular}
\end{center}
\begin{center}
\caption{Information sets for R(2,4)}
\end{center}

\end{table}

 The next value for $m$ is $m=6$. In this case, $n=2^6-1=63, r_1=7, r_2=9, a=3$ and $b=2$. So 
 $$
 \begin{array}{rcl}
\Gamma&=&\{(0,0),(1,0),(2,0),(3,0),(4,0),(5,0),(0,1),(1,1),(2,1),(3,1),(4,1),(5,1),(0,2),(1,2),(2,2),(3,2),
(4,2),\\
&&(5,2),(0,3),(1,3),(2,3)\}.\\
 \end{array}$$
By taking as isomorphism that given by the Chinese Remainder Theorem, we obtain 
$$T^{-1}(\Gamma)=\{0,1,2, 9,10,11,18,19,21,28,29, 30,36,37,38,45,46,47,54,56,57\},$$ so we have that $\{0,\alpha^i\mid i\in T^{-1}(\Gamma)\}$ is an information set for $R(2,6)$ and $\{\alpha^i\mid i\notin T^{-1}(\Gamma)\}$ is an information set for $R(3,6)$.  
 
 Finally, we see the case $m=8$. In this case, there are two decompositions of $n=2^8-1=255$ that satsify conditions (\ref{restrictions}), namely, $(r_1=3, r_2=85)$ and $(r_1=15, r_2=17)$. Table V shows the sets $T^{-1}(\Gamma)$ obtained for each decomposition; in both cases we are considering the isomorphism given by the Chinese Remainder Theorem.

\begin{table}[h]
\begin{center}
\begin{tabular}{|c|c|c|c|}
 \hline
 $r_1$&$r_2$&a& $T^{-1}(\Gamma)$\\ \hline
 3&85&2&\{ 0, 1, 2, 3, 4, 5, 6, 7, 8, 9, 12, 15, 85, 86, 87, 88, 89, 90, 91, 92, 93,\\
 &&&94, 96, 99, 170, 171, 172, 173, 174, 175, 176, 177, 178, 179, 180, 183 \}\\ \hline
 15&17&4&\{  0, 1, 2, 3, 17, 18, 19, 20, 34, 35, 36, 51, 52, 53, 68, 69, 105, 120, 121,
 122, 136,\\
 &&& 137, 138, 139, 153, 154, 155, 170, 171, 172, 187, 188, 189, 204,
 240, 241\}\\ \hline
\end{tabular}
\end{center}
\begin{center}
\caption{Information sets for R(5,8)}
\end{center}

\end{table}

\end{examples}

 To finish our explanation, we include the following table which shows the suitable values of $m$ up to length 4096. The values $m=2,3,5,7$ yield a prime number for $n=2^m-1$, while the value $m=11$ only admits a decomposition $(r_1=23, r_2=89)$ that does not satisfy the conditions (\ref{restrictions}).

\begin{table}[h]
\begin{center}
\begin{tabular}{|c|c|c|c|c|c|}
 \hline
 $\mathbf{m}$&$\mathbf{n}$&$\mathbf{r_1}$&$\mathbf{r_2}$& $\mathbf{a}$ & $\mathbf{b}$\\ \hline
 4&15&3&5&2&2\\ \hline
 6&63&7&9&3&2\\ \hline
 8&255&3&85&2&4\\ \hline
 8&255&15&17&4&2\\ \hline
 9&511&7&73&3&3\\ \hline
 10&1023&3&341&2&5\\ \hline
 12&4095&7&585&3&4\\ \hline
 12&4095&63&65&6&2\\ \hline
\end{tabular}
\end{center}
\begin{center}
\caption{Parameters for second order RM codes up to length 4096}
\end{center}

\end{table}

\section{Conclusions}

 We have described information sets for Reed-Muller codes of first and second-order respectively. We have seen them as group codes and we have constructed those information sets from their defining sets; that is, from an intrinsical characteristic of this family of algebraic codes. This supposes a relevant difference with regard to other approaches to the same problem, such as those proposed from a geometrical point of view (\cite{KMM,KMM2,Senev}). Moreover, the definition of the information sets turn be very simple in the end and it is given only in terms of their basic parameters.  

It is also pertinent to wonder about the benefits that our results may imply in relation to the applicability of certain decoding algorithms such as the permutation decoding algorithm, which is especially appropriate for abelian codes. In \cite{BS2} we showed a generic study of this algorithm for abelian codes starting from the information sets introduced in \cite{BS}. Therefore, since the present work manage to adapt the results in \cite{BS} to Reed-Muller codes, it is reasonable to believe that the results in \cite{BS2}can also be apply to them. Furthermore, the vision of Reed-Muller codes as affine-invariant codes allows us to search for a PD-set within the group of all affine transformations, beyond the translations. For all these reasons, the application of the permutation decoding algorithm to Reed-Muller codes, starting from the results obtained in this work instead of the geometric point of view, supposes an interesting open problem.

%\section{The case $m=ab$ with $a=2$}\label{a=2}
%
%We are interested in the particular case $a=2$, that is, the RM codes such that $m$ even and $(m/2,3)=1$.

\end{document}